\documentclass[journal,onecolumn]{IEEEtran}

\usepackage{amsmath, amsthm, amssymb}
\usepackage{latexsym}
\usepackage{graphicx}
\usepackage{verbatim}
\usepackage{mathrsfs}
\usepackage{float}
\usepackage[lofdepth, lotdepth]{subfig}
\usepackage{array}
\newcolumntype{P}[1]{>{\centering\arraybackslash}p{#1}}

\usepackage{url}
\usepackage{algorithm}
\usepackage[noend]{algorithmic}
\usepackage{cite}
\usepackage[table]{xcolor}
\usepackage{enumerate}
\usepackage{eucal}
\usepackage{stmaryrd}
\usepackage{mathtools}
\usepackage{pgf}
\usepackage{tikz}
\usetikzlibrary{arrows,automata}
\usepackage[latin1]{inputenc}
\usetikzlibrary{automata,positioning}

\usepackage[top=0.75in, bottom=1in, left=0.625in, right=0.625in]{geometry}
\usepackage{mathtools}
\DeclarePairedDelimiter\ceil{\lceil}{\rceil}
\DeclarePairedDelimiter\floor{\lfloor}{\rfloor}

\theoremstyle{plain}
\newtheorem{theorem}{Theorem}
\newtheorem{proposition}{Proposition}
\newtheorem{lemma}{Lemma}

\newtheorem{corollary}{Corollary}
\newtheorem{construction}{Construction}
\theoremstyle{definition}
\newtheorem{definition}{Definition}
\newtheorem{example}{Example}
\newtheorem{remark}{Remark}

\newcommand{\B}{{\mathcal B}}
\newcommand{\C}{{\mathcal C}}
\newcommand{\D}{{\mathcal D}}

\DeclareMathAlphabet{\mathbfsl}{OT1}{ppl}{b}{it} 

\newcommand{\ba}{{\mathbfsl a}}
\newcommand{\bb}{{\mathbfsl b}}

\newcommand{\bu}{{\mathbfsl u}}
\newcommand{\bv}{{\mathbfsl v}}

\newcommand{\by}{{\mathbfsl y}}

\newcommand{\bc}{{\mathbfsl c}}

\newcommand{\bx}{{\mathbfsl{x}}}
\newcommand{\bz}{{\mathbfsl{z}}}

\newcommand{\bsg}{{\boldsymbol{\sigma}}}

\newcommand{\bbZ}{{\mathbb Z}}

\newcommand{\ppmod}[1]{~({\rm mod~}#1)}

\renewcommand{\ge}{\geqslant}
\renewcommand{\le}{\leqslant}

\newcommand{\et}{{\emph{et al.}}}

\newcommand{\enc}{\textsc{Enc}}
\newcommand{\dec}{\textsc{Dec}}

\newcommand{\Bindel}{{\cal B}^{\rm indel}}

\IEEEoverridecommandlockouts
\begin{document}

\pagestyle{plain}




\title{A New Version of $q$-ary Varshamov-Tenengolts Codes with more Efficient Encoders: The Differential VT Codes and The Differential Shifted VT Codes  \\[2mm]}

\author{
   \IEEEauthorblockN{
   	Tuan Thanh Nguyen, \IEEEmembership{Member, IEEE,}
	Kui Cai, \IEEEmembership{Senior Member, IEEE,}
	and Paul H. Siegel, \IEEEmembership{Life Fellow, IEEE}}
\thanks{This work was presented in part at the IEEE 2023 IEEE International Conference on Communications: SAC Cloud Computing, Networking and Storage Track (IEEE ICC)\cite{iccversion}. The work of Tuan Thanh Nguyen and Kui Cai is supported by the Singapore Ministry of Education Academic Research Funds Tier 2 MOE2019-T2-2-123 and T2EP50221-0036. The work of Paul Siegel is supported in part by NSF Grant CCF-2212437.}
\thanks{Tuan Thanh Nguyen and Kui Cai are with the Science, Mathematics, and Technology Cluster, Singapore University of Technology and Design, Singapore 487372 (email: \{tuanthanh\_nguyen, cai\_kui\}@sutd.edu.sg).}
\thanks{Paul H. Siegel is with the University of California, San Diego, La Jolla, CA 92093, USA (email: psiegel@ucsd.edu).}
}

\maketitle

\hspace{-3mm}\begin{abstract}

The problem of correcting deletions and insertions has recently received significantly increased attention due to the DNA-based data storage technology, which suffers from deletions and insertions with extremely high probability. In this work, we study the problem of constructing non-binary burst-deletion/insertion correcting codes. Particularly, for the quaternary alphabet, our designed codes are suited for correcting a burst of deletions/insertions in DNA storage.   

Non-binary codes correcting a single deletion or insertion were introduced by Tenengolts [1984], and the results were extended to correct a fixed-length burst of deletions or insertions by Schoeny \et{} [2017]. Recently, Wang \et{} [2021] proposed constructions of non-binary codes of length $n$, correcting a burst of length at most two for $q$-ary alphabets with redundancy $\log n+O(\log q \log \log n)$ bits, for arbitrary even $q$. The common idea in those constructions is to convert non-binary sequences into binary sequences, and the error decoding algorithms for the $q$-ary sequences are mainly based on the success of recovering the corresponding binary sequences, respectively. 

In this work, we look at a natural solution that the error detection and correction algorithms are performed directly over $q$-ary sequences, and for certain cases, our codes provide a more efficient encoder with lower redundancy than the best-known encoder in the literature. Particularly, 

\begin{itemize}
\item {\bf (Single-error correction codes)} We first present a new version of non-binary VT codes that are capable of correcting a single deletion or single insertion, providing an alternative simpler and more efficient encoder of the construction by Tenengolts [1984]. Our construction is based on the {\em differential vector}, and the codes are referred to as the {\em differential VT codes}. In addition, we provide linear-time algorithms that encode user messages into these codes of length $n$ over the $q$-ary alphabet for $q \ge 2$ with at most $\ceil{\log_q n}+1$ redundant symbols, while the optimal redundancy required is at least $\log_q n+\log_q (q-1)$ symbols. Our designed encoder reduces the redundancy of the best-known encoder of Tenengolts [1984] by at least  $2$ redundant symbols or equivalently $2\log_2 q$ bits. 

\item {\bf (Burst-error correction codes)} We use the idea of the {\em binary shifted VT codes} to define the {\em $q$-ary differential shifted VT codes}, and propose non-binary codes correcting a burst of up to two deletions (or two insertions) with redundancy $\log n+3\log \log n+ O(\log q)$ bits, which improves a recent result of Wang \et{} [2021] with redundancy $\log n+O(\log q \log \log n)$ bits for all $q\ge 8$. We then extend the construction to design non-binary codes correcting a burst of either exactly or at most $t$ deletions (or insertions) for arbitrary $t\ge 2$. 
\end{itemize}

\end{abstract}

\section{Introduction}

Codes correcting deletions and insertions are important for many data storage systems such as the bit-patterned media magnetic recording systems \cite{a1} and racetrack memory devices \cite{a2}.  Insertions and deletions may also occur due to the synchronization errors in communication systems \cite{a3} and mobile data \cite{a4}. Furthermore, the problem of correcting such errors has recently received significantly increased attention due to the DNA-based data storage technology, which suffers from deletions and insertions with extremely high probability \cite{O:2015, Heckel:2019, Nguyen:2021, TT:special,ryan:2022}. Designing codes for correcting deletions and/or insertions is well-known to be a challenging problem, even in the most fundamental settings with only a single error. One of the challenges that make deletions or insertions more destructive than substitutions is that only a small number of errors can cause the original data sequences and the received sequences to be vastly different under the Hamming metric. 

In this work, we focus on the design of non-binary codes that are capable of correcting a burst of deletions (or insertions), where a burst refers to a block of errors that occur in consecutive symbols. This has been pointed out as a typical type of error that arises in DNA-based data storage technology that uses nanopore sequencing technologies \cite{schoenyDNA,O:2013}. In addition, in wireless communications, burst errors also occur with high frequency due to multi-path fading \cite{ye:2008,rao:1999}. In this work, not only are we interested in constructing large error-correction codes, we desire efficient encoders and decoders that map arbitrary user data into these codes and vice versa. In general, code design takes into account the lowest redundancy required to correct such errors with fast encoding and decoding procedures. In this work, we define $\B_t(\bx)$ to be the set of sequences that can be obtained from $\bx$ via a burst of either $t$ deletions or $t$ insertions. Similarly, $\B_{\le t}(\bx)$ is the set of sequences that can be obtained from $\bx$ via a burst of at most $t$ deletions or at most $t$ insertions.  

Over the $q$-ary alphabet, $q\ge 2$, consider a channel model with a given error ball function $\B$, and suppose that the optimal redundancy required to correct such errors is ${\rm r}_{q,\B}$, then two crucial coding theory problems are: 
\vspace{1mm}

\noindent {\bf P1: Code Design.} Can one design the largest possible code $\C$, with the redundancy ${\rm r}_{\C}$, such that ${\rm r}_{\C} \to {\rm r}_{q,\B}$? 
\vspace{1mm}

\noindent {\bf P2: Encoder/Decoder Design.} Can one design an efficient encoder $\enc$ (and a corresponding decoder $\dec$) that encodes arbitrary user messages into codewords in $\C$ with nearly-optimal redundancy ${\rm r}_{\enc}$, ${\rm r}_{\enc} \to {\rm r}_{\C}$?
\vspace{1mm}

In the literature, the problems of constructing codes (problem P1) correcting a burst of exactly $t$ deletions (or exactly $t$ insertions), also known as {\em fixed-length burst}, and a burst of at most $t$ deletions (or at most $t$ insertions), also known as {\em variable-length burst} have both been studied, with the latter being the more complex problem \cite{WANG:BURST,binary-burst,VT:1965,Le:1965,kas:1998,Tene:1984,del-in,le:2del,2del:2,len:burst,cheng-burst}. On the other hand, designing efficient encoders (problem P2) is crucial for practical applications, however, in many settings, it remains an open challenge, even in the most fundamental settings with only a single error.  
\vspace{1mm}

\noindent {\bf Non-binary single-error correction codes}. The first challenge comes from extending the coding solutions in binary codes to non-binary codes. Particularly, while the problems of giving nearly-optimal explicit constructions of codes (P1) and designing nearly-optimal encoders for such codes (P2) over the binary alphabet have been settled for more than 50 years, the approach fails to be extended to the case of $q$-ary alphabet for any fixed $q>2$. 
In particular, to correct a single deletion or single insertion, we have the celebrated class of Varshamov-Tenengolts (VT) codes. In 1965, Varshamov and Tenengolts introduced the binary VT codes to correct asymmetric errors \cite{VT:1965}, and Levenshtein subsequently showed that such codes can be used for correcting a deletion or insertion with a simple linear-time decoding algorithm \cite{Le:1965}. For codewords of length $n$, the binary VT codes incur $\log (n+1)$ redundant bits \footnote{In this work, for simplicity, we use the notation ``$\log$" without the base to refer to the logarithm of base two.}
, while the optimal redundancy, provided in \cite{Le:1965}, is at least $\log n$ bits. Curiously, even though the binary VT codes and efficient decoding algorithm were known since 1965, a linear-time encoder for such codes was only proposed by Abdel-Ghaffar and Ferriera in 1998 \cite{kas:1998}, which used $\ceil{\log (n+1)}$ redundant bits. We observe that, over the binary alphabet, (P1) and (P2) are solved asymptotically optimal: 
\begin{equation*}
{\rm r}_{2,\B_1} \ge \log n, \text{ } {\rm r}_{\C}=\log (n+1), \text{and } {\rm r}_{\enc}=\ceil{\log (n+1)}.
\end{equation*}

For the non-binary alphabet, in 1984, a non-binary version of the VT codes was proposed by Tenengolts \cite{Tene:1984}, and the constructed codes can correct a single deleted or inserted symbol in the $q$-ary alphabet with a linear-time decoder for any $q>2$. The construction of Tenengolts retains the attractive properties of the binary VT codes, such as the simple decoding algorithm. For codewords of length $n$, such codes incur at most $\log_q n +1$ redundant symbols. In the same paper, Tenengolts also provided an upper bound for the cardinality of any $q$-ary codes of length $n$ correcting a deletion or insertion, which is at most $q^n/((q-1)n)$, and hence, the minimum redundancy required is at least $\log_q n+\log_q (q-1)$ symbols. Unlike the binary case, designing an efficient encoder that encodes arbitrary user messages into Tenengolts' code is a challenging task (refer to Section III-A for detailed discussion). To overcome the challenge, several attempts have been made in three variations:
\begin{itemize}
\item {\em Targeting a specific value of $q$}. When $q=4$, 
Chee \et{} \cite{chee:2019} presented a linear-time quaternary encoder that corrects a single deletion or insertion with $\ceil{\log_4 n}+1$ redundant symbols. The redundancy is asymptotically optimal. Unfortunately, the approach fails to be extended to the case of $q$-ary alphabet for arbitrary $q>2$. 
\item {\em Using more redundancy.} Abroshan \et{} \cite{M:2018} presented a systematic encoder that maps user messages into a single $q$-ary VT code as constructed in \cite{Tene:1984} with complexity that is linear in the code length. Unfortunately, the redundancy of this encoder is more than $\log_q n+\log n$ symbols (see Section II). 
\item {\em Relaxing the condition for output codewords}. In \cite{Tene:1984}, Tenengolts provided a systematic encoder that requires at least $\ceil{\log_q n}+3$ symbols, which is the best-known encoder for codes that correct a single deletion or insertion. In term of redundancy, a natural question is: can one construct a linear-time encoder with at most $r$ redundant symbols, where $\log_q n+\log_q (q-1) \le r < \ceil{\log_q n}+3$? In addition, The drawback of the encoder in \cite{Tene:1984} is that the codewords obtained from this encoder are not contained in a single $q$-ary VT code. Note that to correct a single deletion or insertion, it is not necessary that all the codewords must belong to the same coset of $q$-ary VT codes. Nevertheless, when the words share the same parameters, Abroshan \et{} \cite{M:2018} demonstrated that these codes can be adapted to correct multiple insertion/deletion errors, in the context of {\em segmented edits}  \cite{M:segment, Liu:2010, Cai:segment}.   
\end{itemize}

\noindent {\em Our contribution for single-error correction codes.} Motivated by the code design problem above, we present a new version of non-binary VT codes that give asymptotically optimal solutions for (P1) and (P2), as best as over binary alphabet, as follows:
 \begin{equation*}
{\rm r}_{q,\B_1} \ge \log_q n+\log_q(q-1), {\rm r}_{\C}=\log_q n +1, \text{and } {\rm r}_{\enc}=\ceil{\log_q n}+1.
\end{equation*}
\noindent Our construction is based on the {\em differential vector}, and the codes are referred to as the {\em differential VT codes}. Our constructed codes have the same cardinality and redundancy, as compared to the best known $q$-ary single deletion/insertion codes constructed by Tenengolts \cite{Tene:1984}. On the other hand, our proposed code construction method supports more efficient encoding and decoding procedures (in other words, it enables an easier method to solve (P2)). Consequently, our best encoder uses at most $\ceil{\log_q n}+1$ redundant symbols, and hence, it reduces the redundancy of the best known encoder of Tenengolts \cite{Tene:1984} by at least $2$ redundant symbols, or equivalently $2\log q$ redundant bits.
\vspace{1mm}

\begin{table*}[h!]
\centering 
 \begin{tabular}{ |P{2.5cm}|P{2cm}| P{4.5cm}| P{4.5cm}|}
 \hline
 &  Size of burst  &  (P1) Redundancy of the constructed code $\C$  &  (P2) Redundancy of the encoder for $\C$ \\[1ex]
 \hline
{\color{black}{Tenengolts \cite{Tene:1984}}} & {\color{black}{$=1$}} & {\color{black}{$\log_q n+1$ (symbols)}}  & {\color{black}{$\ceil{\log_q n}+3$ (symbols)}} \\
 \hline
{\bf {\color{black}{This work}}} & {\color{black}{$=1$}} & {\color{black}{$\log_q n+1$ (symbols)}}  & {\color{black}{$\ceil{\log_q n}+1$ (symbols)}}\\
 \hline
 {\color{black}{Wang \et{} \cite{2del:2}}} & {\color{black}{$\le2$}} &   {\color{black}{$\log n+O(\log q \log \log n)$ (bits)}}   &  {\color{black}{$\log q \log n + O(\log q)$ (bits)}}\\
 \hline
{\bf {\color{black}{This work}}} & {\color{black}{$\le2$}} &  {\color{black}{$\log n+3 \log \log n+O(\log q)$ (bits)}}  &  {\color{black}{$\log n+3 \log \log n+O(\log q)$ (bits)}}\\
  \hline
  {\color{black}{Schoeny \et{} \cite{schoenyDNA}}} & {\color{black}{$=t$}} &   {\color{black}{$\log n+ (t-1) \log \log n + O(t\log q)$ (bits)}}   &  {\color{black}{NA}}\\
 \hline
{\bf {\color{black}{This work}}} & {\color{black}{$=t$}} &  {\color{black}{$\log n+ (t-1) \log \log n + O(t\log q)$ (bits)}}   &  {\color{black}{$\log n+ (t-1) \log \log n + O(t\log q)$ (bits)}} \\
  \hline
  {\color{black}{Wang \et{} \cite{WANG:BURST}}} & {\color{black}{$\le t$}} &   {\color{black}{$\log n+ O(\log q \log \log n)$ (bits)}}   &  {\color{black}{NA}}\\
 \hline
{\bf {\color{black}{This work}}} & {\color{black}{$\le t$}} &  {\color{black}{$\log n+  O(t^2\log \log n) + O(t\log q)$ (bits)}}  &  {\color{black}{NA}}\\
  \hline
\end{tabular}
\caption{Related works for non-binary codes in the literature and the main contributions of this work.} 
\label{table1}
\end{table*}

\noindent {\bf Non-binary burst-error correction codes}. The earliest work on the subject, proposed by Levenshtein in 1967 \cite{le:2del}, provided an efficient construction of binary codes capable of correcting a burst at most two deletions (or two insertions) that had redundancy $\log n + 1$ for codewords of length $n$. Binary codes correcting a burst of deletions (or insertions) were later proposed in \cite{binary-burst,len:burst}. Particularly, for an arbitrary constant $t>1$, Schoeny \et{} \cite{binary-burst} proposed binary codes correcting a burst of length exactly $t$, while the work of Lenz and Polyanskii in \cite{len:burst} can correct a burst of variable length up to $t$. Note that, there is a significant difference between codes that can correct a burst of length at most
$t$ and a burst of length exactly $t$, as a code of the earlier type can correct errors of the latter, but the converse is not true in general.
Over the general $q$-ary alphabet, recently, Wang \et{} \cite{2del:2} proposed constructions of codes of length $n$, correcting a burst of length at most two with redundancy $\log n+O(\log q \log \log n)$ bits, for arbitrary even $q$. The results were later extended to construct non-binary codes correcting a burst of up to $t$ deletions (or insertions) in \cite{WANG:BURST}. However, designing efficient encoders (problem P2) for such constructed codes remains an open challenge, even in the case of $t = 2$. Particularly, to correct a burst of at most 2 errors, the authors \cite{WANG:BURST} provided a systematic construction of encoder, however, the redundancy is roughly $\log q \log n + O(\log q)$, which is much larger than the constructed codes whose redundancy was only $\log n+O(\log q \log \log n)$ bits. 
\vspace{1mm}

\noindent {\em Our contribution for burst-error correction codes.} We use the idea of the {\em binary shifted VT codes} to define the {\em $q$-ary differential shifted VT codes}, which is crucial to the construction of $q$-ary codes correcting a burst of errors. Given $t>0$, we propose non-binary codes correcting a burst of either exactly or at most $t$ deletions/insertions. Particularly, for $t=2$ and a given $q$-ary alphabet, we construct non-binary codes of length $n$ that can correct a burst of at most two deletions or two insertions with redundancy $\log n+3\log \log n+ O(\log q)$ bits, which improves a recent result of Wang \et{} [2021] with redundancy $\log n+O(\log q \log \log n)$ bits for all $q\ge 8$. In addition, we present a linear-time encoder that encodes arbitrary user messages into non-binary codes correcting a burst of at most two deletions with redundancy $\log n+3\log \log n+ O(\log q)$ bits, which improves the redundancy of the encoder in \cite{WANG:BURST}. 
\vspace{1mm}

The remainder of this paper is organized as follows. We first go through notations and some preliminary results in Section II. In Section III-A, we focus on the single error correction code, i.e. $t=1$, and present a new version of non-binary VT codes, which are referred to as the {\em differential VT codes}. In addition, in Section III-B, we present a linear-time encoder that encodes user messages into the codes, and for codewords of length $n$ over the $q$-ary alphabet, our designed encoder uses at most $\ceil{\log_q n} + 1$ redundant symbols. The efficiency of our proposed encoders, compared to previous works on single error correction codes, is illustrated in Table II. In Section IV, we introduce the {\em differential shifted VT codes} and propose non-binary codes correcting a burst of exactly $t$ errors with redundancy $\log n+(t-1)\log \log n+O(t\log q)$ bits, and design linear-time encoders for such codes. We then extend the coding method to correct at most $t$ deletions in Section V. Particularly, when $t=2$, our codes incur $\log n+3\log \log n+ O(\log q)$ bits, which improves a recent result of Wang \et{} [2021] with redundancy $\log n+O(\log q \log \log n)$ bits for all $q\ge 8$. Finally, Section VI concludes the paper. A summary of our contributions is illustrated in Table~\ref{table1}. 
\vspace{1mm}

\section{Preliminary}\label{sec:prelim}

Let $\Sigma_q$ denote an {\em alphabet} of size $q$, where $\Sigma_q=\{0,1,2,\ldots, q-1\}$.
For any positive integer $m<n$, we let $[m,n]$ denote the set $\{m,m+1,\ldots,n\}$ and $[n]=[1,n]$.

Given two sequences $\bx$ and $\by$, we let $\bx\by$ denote the {\em concatenation} of the two sequences.
In the special case where $\bx,\by \in \Sigma_q^n$, we use $\bx || \by$ to denote their {\em interleaved sequence} $x_1y_1x_2y_2\ldots x_ny_n$.
For a subset $I=\{i_1,i_2,\ldots, i_j\}$ of coordinates, we use $\bx|_I$ to denote the vector $x_{i_1}x_{i_2}\ldots x_{i_j}$. A sequence $\by$ is said to be a {\em subsequence} of $\bx$, if there exists a subset of coordinates $I$ such that $\by=\bx|_I$. We now introduce the definition of a burst of deletions or insertions.

\begin{definition}
Given $\bx=(x_1,x_2\ldots, x_n)\in \Sigma_q^n$. We say that $\bx$ suffers a burst of $t$ deletions if exactly $t$ consecutive symbols have been deleted from $\bx$, resulting a subsequence $\bx'=(x_1, x_2,\ldots, x_i, x_{i+t+1}, x_{i+t+2}, \ldots, x_n) \in \Sigma_q^{n-t}$ for some $i\in[n-t]$. On the other hand, we say that $\bx$ suffers a burst of $t$ insertions if exactly $t$ consecutive insertions have occurred from $\bx$, resulting a subsequence $\bx''=(x_1, x_2,\ldots, x_j,{\color{red}{y_1,y_2,\ldots,y_t}}, x_{i+1}, x_{i+2}, \ldots, x_n) \in \Sigma_q^{n+t}$ for some $i\in[n]$. Similarly, we say $\bx$ suffers a burst of up to $t$ deletions if $s_1$ consecutive symbols have been deleted for some $s_1\le t$, or $\bx$ suffers a burst of up to $t$ insertions if $s_2$ consecutive insertions have occurred for some $s_2\le t$.
\end{definition}

 In this work, we define $\B_t(\bx)$ to be the set of sequences that can be obtained from $\bx$ via a burst of either $t$ deletions or $t$ insertions. Similarly, $\B_{\le t}(\bx)$ is the set of sequences that can be obtained from $\bx$ via a burst of at most $t$ errors. 
 
\begin{definition}
Let $\C\subseteq \Sigma_q^n$. We say that $\C$ corrects a burst of $t$ deletions or $t$ insertions if 
and only if $\B_t(\bx)\cap \B_t(\by)=\varnothing$ for all distinct $\bx,\by \in \C$. Similarly, we say that $\C$ can correct a burst of up to $t$ deletions or up to $t$ insertions if 
and only if $\B_{\le t}(\bx)\cap \B_{\le t}(\by)=\varnothing$ for all distinct $\bx,\by \in \C$.
\end{definition}

For a code $\C \subseteq \Sigma_q^n$, the redundancy is measured by the value ${\rm r}_{\C}=n-\log_q |\C|$ (in symbols) or $n \log q-\log |\C|$ (in bits). In this work, not only are we interested in constructing large error-correction codes (problem P1), we desire an efficient encoder that maps arbitrary user data into these codes (problem P2). 

\begin{definition}
The map $\enc: \Sigma_q^k\to \Sigma_q^n$
is a {\em $t$-burst-encoder} if there exists a {\em decoder} map $\dec:\Sigma_q^{n+t}\cup \Sigma_q^n \cup \Sigma_q^{n-t} \to \Sigma_q^n$ such that the following conditions hold:
\begin{itemize} 
\item For all $\bx\in\Sigma_q^k$, we have $\dec\circ\enc(\bx)=\bx$,
\item If $\bc=\enc(\bx)$ and $\bc'\in \B_t(\bc)$, then $\dec(\bc')=\bx$.
\end{itemize}
Hence, we have that the code $\C=\{\bc : \bc=\enc(\bx),\, \bx\in\Sigma_q^k\}$ and $|\C|=q^k$.
The {\em message length} is $k$ while the {\em codeword length} is $n$. The {\em redundancy of the encoder} is measured by the value $n-k$ (in symbols) or $(n-k)\log q$ (in bits). A {\em ${{\le}t}$-burst-encoder} can be defined similarly.
\end{definition}

\begin{definition} For $q\ge 2$, the {\em VT syndrome} of a $q$-ary sequence $\bx\in\Sigma_q^n$ is defined to be
${\rm Syn}(\bx)=\sum_{i=1}^n i x_i$.
\end{definition}

To correct a single deletion or single insertion, we have the celebrated class of Varshamov-Tenengolts (VT) codes. 


\begin{construction}[Binary VT codes \cite{VT:1965}]\label{cons1} 
Given $n>0$ and $q=2$. For $a \in  \bbZ_{n+1}$, let
\begin{equation*}\label{VTcodes}
{\rm VT}_a(n)=\Big\{\bx\in \{0,1\}^n: {\rm Syn}(\bx) = a \ppmod{n+1}\Big\}.
\end{equation*}  
\end{construction}

\begin{theorem}[Levenshtein, 1965 \cite{Le:1965}]
For $a \in  \bbZ_{n+1}$, ${\rm VT}_a(n)$ can correct a single deletion or a single insertion. There exists $a\in\bbZ_{n+1}$ such that ${\rm VT}_a(n)$ has at least $2^n/(n+1)$ codewords, and the redundancy of the code is at most $\log (n+1)$ bits.
\end{theorem} 


Over the nonbinary alphabet, in 1984, Tenengolts \cite{Tene:1984} generalized the binary VT codes to $q$-ary VT codes for any fixed $q$-ary alphabet. Crucial to the construction of Tenengolts in \cite{Tene:1984} was the concept of the {\em signature vector} defined as follows. 

\begin{definition} The {\em signature vector} of a $q$-ary vector $\bx$ of length $n$ 
is a binary vector $\alpha(\bx)$ of length $n-1$, 
where $\alpha(x)_i=1$ if $x_{i+1}\geq x_i$, and $0$ otherwise, for $i\in[n-1]$.
\end{definition}

\begin{construction}[$q$-ary VT codes as proposed in \cite{Tene:1984}]\label{cons2} 
Given $n,q>0$, for $a\in \bbZ_n$ and $b\in \bbZ_q$, set 
{
\begin{align*}
\label{qaryVT}
    {\rm T}_{a,b}({n;q}) \triangleq \Big\{ \bx \in \bbZ_q^n : &\alpha(\bx)\in{\rm VT}_a(n-1) \text{ and } \sum_{i=1}^n x_i = b\ppmod{q} \Big\}.
\end{align*}  
}
\end{construction}

\begin{theorem}[Tenengolts, 1984 \cite{Tene:1984}]
The set ${\rm T}_{a,b}(n;q)$ forms a $q$-ary single deletion/insertion correction code and 
there exists $a$ and $b$ such that the size of ${\rm T}_{a,b}({n;q})$ is at least $q^n/(qn)$. There exists a systematic encoder $\enc_{\rm T}$
with redundancy $\ceil{\log n} + 3\ceil{\log q} $ (bits) or $\ceil{\log_q n} + 3$ (symbols).
\end{theorem} 

On the other hand, the codewords obtained from the encoder $\enc_{\rm T}$ are not contained in a single $q$-ary VT code ${\rm T}_{a,b}(n;q)$. Recently, Abroshan \et{} \cite{M:2018} presented a systematic encoder that maps binary messages into ${\rm T}_{a,b}(n;q)$. Unfortunately, the redundancy of the encoder is as large as $\log n(\log q+1)+2(\log q-1)$ bits, and hence, more than $\log n+\log_q n$ symbols.

\section{Correcting a Single Deletion or Insertion: a New Version of $q$-ary VT Codes}

{\em A Natural Idea from Binary VT Codes.} Recall the design of the binary VT codes ${\rm VT}_a(n)$ from Construction 1 to correct a single deletion or insertion. A natural question is whether there exists a simple VT syndrome over $q$-ary codewords to correct single deletion or insertion for arbitrary $q>2$. Observe that, in the construction of Tenengolts \cite{Tene:1984} (refer to Construction 2), the VT syndrome is enforced over the signature of each codeword, which is a binary sequence. That is a drawback leading to the difficulty of designing an efficient encoder as in the binary case. Consequently, to encode arbitrary messages into ${\rm T}_{a,b}(n;q)$ by enforcing the VT syndrome over the binary signature sequences, Abroshan \et{} \cite{M:2018} required more than $\log_q n+\log n$ redundant symbols. A natural solution should be obtained by enforcing a single VT syndrome over all $q$-ary sequences, and consequently, the design of a corresponding encoder would be simple as in the binary case. On the other hand, we observe that imposing VT syndrome directly over every $q$-ary codeword is not sufficient to correct a deletion or insertion. For example, it is easy to verify that the following two sequences $\bz_1 = \bx213\by$ and $\bz_2 = \bx132\by$, where $\bx, \by$ are arbitrary sequences, have the same VT syndrome, however, share a common sequence in the single error ball as $\bz' = \bx13\by$. The first contribution of our work is to show that imposing the VT syndrome over the {\em differential vector} of every $q$-ary codeword allows us to correct a single error.   


\subsection{The Differential VT Codes} 

\begin{definition}
Given $\bx \in \Sigma_q^n$. The {\em differential vector} of $\bx$, denoted by ${\rm Diff}(\bx)$, is a sequence $\by={\rm Diff}(\bx) \in \Sigma_q^n$ where: 
\begin{equation*}
\left\{ \begin{array}{ll}
y_i &=x_i - x_{i+1} \ppmod{q} \mbox{, for } 1\le i\le n-1, \\ 
y_n &=x_n.
\end{array}\right.
\end{equation*}
\end{definition}

Clearly, ${\rm Diff}(\bx)$ is a one-to-one function. From $\by={\rm Diff}(\bx)$, we can obtain $\bx= {\rm Diff}^{-1}(\by)$ as follows.
\begin{equation*}
\left\{ \begin{array}{ll}
x_n &= y_n, \mbox{ and }\\ 
x_i &= \sum_{j=i}^n y_j \ppmod{q} \mbox{, for } n-1\ge i\ge 1.
\end{array}\right.
\end{equation*}

\begin{construction}[{\bf The $q$-ary  Differential VT codes}]
Given $n>0$. For $q\ge 2, a\in \bbZ_{qn}$, set
\begin{equation*} 
{\rm Diff\_VT}_{a}({n;q}) \triangleq \big\{ \bx \in \Sigma_q^n: {\rm Syn}({\rm Diff}(\bx)) = a  \ppmod{qn} \big\}.
\end{equation*}
\end{construction} 

Our main contribution in this section is summarized as follows. 

\begin{theorem}\label{mainresult}
The code ${\rm Diff\_VT}_{a}({n;q})$ can correct a single deletion or single insertion in linear time. In other words, there exists a linear-time decoder $\dec_{\rm error}: \Sigma_q^{n-1} \cup \Sigma_q^{n+1} \to \Sigma_q^n$ such that if $\bx'$ is obtained from $\bx \in {\rm Diff\_VT}_{a}({n;q})$ after a deletion or an insertion, we can recover $\bx=\dec_{\rm error}(\bx')$. In addition, there exists $a\in \bbZ_{qn}$, such that $\big|{\rm Diff\_VT}_{a}({n;q}) \big| \ge {q^n}/{(qn)}$.
\end{theorem}

The following lemmas are crucial to show the correctness of Theorem~\ref{mainresult}.

\begin{lemma}\label{trans-lemma1}
Given $\bx \in \Sigma_q^n$ and let $\by={\rm Diff}(\bx) \in \Sigma_q^n$. Suppose that $\bx'$ is obtained via $\bx$ by a deletion at symbol $x_i$ for $1\le i\le n$. We then have:
\begin{enumerate}[(i)]
\item If $2\le i\le n$, then $y_{i-1} y_{i}$ is replaced by $y_{i-1} + y_i \ppmod{q}$,  
\item If $i=1$, then $y_1$ is deleted in ${\rm Diff}(\bx)$. 
\end{enumerate}
\end{lemma} 

\begin{proof}
We have $\by={\rm Diff}(\bx)$, where $y_i=x_i - x_{i+1} \ppmod{q}$ for $1\le i\le n-1$ and $y_n=x_n$. 
\vspace{2mm}

If $i=1$, i.e. $x_1$ is deleted in $\bx$, we then have $\bx'=x_2x_3\ldots x_n$. Clearly, ${\rm Diff}(\bx')=y_2y_3\ldots y_n$, or $y_1$ is deleted in ${\rm Diff}(\bx)$. 
\vspace{2mm}

If $2\le i\le n$, a deletion at $x_i$ affects $y_{i-1}, y_i$ in ${\rm Diff}(\bx)$, as $y_{i-1}=x_{i-1}-x_i \ppmod{q}$ and $ y_i=x_i-x_{i+1} \ppmod{q}$. We observe that the change in ${\rm Diff}(\bx')$ is then 
\begin{align*}
{\rm Diff}(\bx')_{i-1} &= x_{i-1}- x_{i+1} \ppmod{q} \\ 
 &= (x_{i-1}-x_{i})+ (x_{i}- x_{i+1})\ppmod{q} \\ 
 &= y_{i-1} + y_i \ppmod{q}.
\end{align*} 
We conclude that $y_{i-1} y_{i}$ is replaced by $y_{i-1} + y_i \ppmod{q}$. 
\end{proof}

\begin{example}
Consider  $\Sigma_4=\{0,1,2,3\}$, and $\bx=0{\color{red}{2}}11301$. We then have  $\by={\rm Diff}(\bx)= {\color{blue}{2 1}} 0 2 3 31$. Suppose that the symbol 2 is deleted in $\bx$, resulting $\bx'=011301$, and ${\rm Diff}(\bx')={\color{green}{3}}02331$. In this example, we observe that, ${\color{red}{x_2}}$ is deleted in $\bx$, and the resulting ${\color{blue}{y_{1} y_{2}=21}}$ in ${\rm Diff}(\bx)$ is replaced by ${\color{green}{3}}=y_{1} + y_2$.  
\end{example}


\begin{lemma}[Parity check lemma]\label{sumsymbol}
Given $n>0$, $q\ge 2,$ and $a\in \bbZ_{qn}$. Consider $\bx\in \Sigma_q^n$ such that ${\rm Syn}({\rm Diff}(\bx)) = a  \ppmod{qn}$. We then have $\sum_{i=1}^n x_i \equiv a \ppmod{q}.$
\end{lemma}
\begin{proof}
Let $\by={\rm Diff}(\bx)$, where $y_i=x_i - x_{i+1} \ppmod{q}$ for $1\le i\le n-1$ and $y_n=x_n$. Suppose that ${\rm Syn}(\by)=a+kqn$ for some positive integer $k$. We have 
\begin{small}
\begin{align*}
{\rm Syn}(\by) &= \sum_{i=1}^{n-1} iy_i + ny_n  \\
&\equiv \sum_{i=1}^n i(x_i - x_{i+1}) + n x_n \ppmod{q} \\
&\equiv \sum_{i=1}^n x_i \ppmod{q}. 
\end{align*}
\end{small}
Since ${\rm Syn}(\by)=a+kqn$, it implies $\sum_{i=1}^n x_i \equiv a \ppmod{q}.$
\end{proof}



We are now ready to show the correctness of Theorem~\ref{mainresult}. Note that any code that corrects $k$ deletions if and only if it can correct $k$ insertions, as established by Levenshtein \cite{del-in}. Also, a code $\C$ can correct a deletion burst of size exactly (or at most) $k$ if and only if it can correct an insertion burst of size exactly (or at most, respectively) $k$ (refer to Theorem 2, Theorem 3 in \cite{binary-burst}). Therefore, for simplicity, throughout this paper, we present the decoding algorithm to correct deletion errors only.
\vspace{1mm}

\noindent{\bf Proof of Theorem~\ref{mainresult}.} Observe that the lower bound is verified by using the pigeonhole principle. It remains to show that the code ${\rm Diff\_VT}_{a}({n;q})$ can correct a single deletion in linear time. 
\vspace{1mm}

For a codeword $\bx\in {\rm Diff\_VT}_{a}({n;q})$, let $\bx'$ be obtained from $\bx$ after a deletion of symbol $\gamma$ at index $i$, i.e. $x_i=\gamma$. According to Lemma~\ref{sumsymbol}, we can obtain the value of the deleted symbol as follows: $\gamma = a-\sum_{j=1}^{n-1} x_j' \ppmod{q}$. 
It remains to determine the value of $i$, i.e. the location of the deleted symbol. Let $\by={\rm Diff}(\bx)$ and $\by' = {\rm Diff}(\bx')$. We then compute: 
\begin{align*}
\Delta  &= {\rm syn}(\by)-{\rm Syn}(\by')=a-{\rm Syn}(\by') \ppmod{qn}, \text{ and}\\
s &=\sum_{j=1}^{n-1} y_j', \text{ i.e. the sum of symbols in } \by'. 
\end{align*}
Observe that the code's parameters such as $a, q, n$ are known, and the received sequence $\bx'$ and its differential vector $\by'$ are known, hence, the values of $\Delta$ and $s$ can be determined. Let $s_R=\sum_{j=i}^n y_j$. We show how $\by$ can be recovered from $\by'$ and thus $\bx$ can be recovered based on $\Delta$ and $s$, which are computable at the decoder. We now have the following cases. 
\vspace{1mm}

\noindent{\bf Case 1.} If $i=1$, we consider a non-trivial case that $y_1> 0$. Indeed, if $y_1=0$, it implies $x_1=x_2$, and such a deletion in $x_1$ is equivalent to a deletion in $x_2$, which is considered in Case 2. Thus, we obtain $\Delta= y_1 + \sum_{j=1}^{n-1} y_j' = y_1+s > s$ and $\Delta< q+s$. 
\vspace{1mm}

\noindent{\bf Case 2.} If $2\le i\le n$, according to Lemma~\ref{trans-lemma1},  $y_{i-1} y_{i}$ is replaced by $y_{i-1} + y_i \ppmod{q}$. 
\begin{itemize}
\item (2a) If $y_{i-1}+y_{i} \le q-1$, then it is easy to verify that $\Delta=y_i+\sum_{j=i}^{n-1} y_j'=s_R \le s$.
\item (2b) If $q\le y_{i-1}+y_{i} \le 2(q-1)$, then we must have $y'_{i-1}+q=y_{i-1}+y_i$. Consequently, we obtain:
\begin{small}
\begin{align*}
\Delta &=(i-1)q+{\color{blue}{y_i}}+\sum_{j=i}^{n-1} y_j'  \\
&= (i-1)q+{\color{blue}{(q-y_{i-1})+y_{i-1}'}}+\sum_{j=i}^{n-1} y_j'  \\
&= (i-1)q+ (q-y_{i-1}) + s - \sum_{j=1}^{i-2} y_j' \\
&= q+s+(q-y_{i-1}) + \Big( (i-2)q - \sum_{j=1}^{i-2} y_j' \Big) \\
&>q+s. 
\end{align*}
\end{small}
\end{itemize}

Therefore, given the computed values $\Delta$ and $s$, we can distinguish all three cases: case 1, case (2a) and case (2b). Moreover, observe that both $s_R$ and $iq+s_R$ are monotonic functions in the index $i$. Particularly, it is easy to verify that $s_R$ is decreasing in the index $i$ while $iq+s_R$ is increasing function in the index $i$. Hence, $\Delta$ is decreasing in the case (2a) while it is increasing in the case (2b). In other words, given the value of $x_i=\gamma$, there is a unique value of $i$ according to the value of $\Delta$. It is easy to see that, in the case when the deleted symbol belongs to a run of identical symbols, we then have more than one option for the index $i$. Nevertheless, we obtain the same codeword. Consequently, to locate the error in $\by$, for (2a), the decoder scans $\by'$ and simply searches for the first index $h$ where $\sum_{j=h}^{n-1} y_j'>\Delta$, while for (2b), the decoder scans $\by'$ and simply searches for the largest index $h$ where $qh+\sum_{j=h}^{n-1} y_j'<\Delta$. The error location in $\bx$ is then $i=h+1$.

In conclusion, the code ${\rm Diff\_VT}_{a}({n;q})$ can correct a single deletion (or equivalently, a single insertion). \qed 

The following result is immediate. 

\begin{corollary}[{\bf The modified $q$-ary  Differential VT codes}]
Given $n,q$. For an arbitrary $N\ge n, a\in \bbZ_{qN}$, set
\begin{equation*} 
{\rm Diff\_VT}^*_{a}({n;q}) \triangleq \big\{ \bx \in \Sigma_q^n: {\rm Syn}({\rm Diff}(\bx)) = a  \ppmod{qN} \big\}.
\end{equation*}
We then have ${\rm Diff\_VT}^*_{a}({n;q})$ is a single deletion/insertion correcting code.  
\end{corollary} 

\begin{remark}
One may construct a code ${\rm Diff\_VT}_{a}({n;q})$ using different variations of the differential function ${\rm Diff}(\bx)$ as follows. For all values $p$, $1\le p\le q-1$ and ${\rm gcd}(p,q)=1$, this coding method works for all {\em $p$-transformation vector} $\Gamma_p(\bx)$, defined as: 
\begin{equation*}
\left\{ \begin{array}{ll}
y_i &=p(x_i - x_{i+1}) \ppmod{q} \mbox{, for } 1\le i\le n-1, \\ 
y_n &=px_n.
\end{array}\right.
\end{equation*}
Another variation of the differential vector was used in \cite{le:2del,2del:2} for binary codes to correct a burst of at most two deletions. 
\end{remark}

\begin{example}
Given $n=10, q=4, a=0$, $\Sigma_4=\{0,1,2,3\}$. Consider a codeword $\bx=0103112013 \in {\rm Diff\_VT}_{0}({10;4})$. We obtain $\by={\rm Diff}(\bx)=3112032323$. It is easy to verify that ${\rm Syn}(\by)=120\equiv 0 \ppmod{40}$ and $\sum_{i=1}^{10} x_i \equiv 0 \ppmod{4}$. 

Suppose that we receive $\bx'=013112013$, i.e. a deletion occurs at $x_3=0$. We then obtain $\by'={\rm Diff}(\bx')=322032323$. 
Now, to correct $\bx$ and find out the value of $i$, we follow the decoding procedure in Theorem 3 as follows. 
\begin{itemize}
\item From $\bx'$, the decoder finds the value of the deleted symbol, which is $a-\sum_{i=1}^{n-1} x_i' = 0 \ppmod{4}$.
\item From $\by'={\rm Diff}(\bx')=322032323$, the decoder computes: 
\begin{small}
\begin{align*} 
\Delta &= a-{\rm Syn}(\by')=0-104=16 \ppmod{40}, \\
s &= \sum_{i=1}^{n-1} y_i'=3+2+2+3+2+3+2+3=20.
\end{align*}
\end{small}
\item Since $\Delta<s$, the decoder concludes that it belongs to the case (2a) where the deletion is not at the first position, i.e. $i\neq 1$, and $y_{i-1}+y_i<q=4$. 

\item Find the error location in $\by$. It can be observed that $\sum_{h=2}^9 y_i' =17 > \Delta=16$ while $\sum_{h=3}^9 y_i' =15 < \Delta$. The decoder then concludes that the error in $\by$ is at the $h=2$ position, and hence, the error in $\bx$ is at $i=h+1=3$.  

\item To correct $\bx$, it inserts the symbol $0$ to the third position. 
\end{itemize}

We now consider another case, where we receive a sequence $\bx'=010311213$, i.e. a deletion occurs at $x_8=0$. We then obtain $\by'={\rm Diff}(\bx')=311203123$. We verify that $y_7y_8=23$ has been replaced to $y_7'=y_2+y_3=1$ in $\by'$. Now, to correct $\bx$ and find out the value of $i$, we follow the decoding procedure in Theorem 3 as follows. 
\begin{itemize}
\item From $\bx'$, the decoder finds the value of the deleted symbol, which is $a-\sum_{i=1}^{n-1} x_i' = 0-(0+1+0+3+1+1+2+1+3)=0 \ppmod{4}$.
\item From $\by'={\rm Diff}(\bx')=311203123$, the decoder computes:
\begin{align*} 
\Delta &= a-{\rm Syn}(\by')=0-84=36 \ppmod{40}, \\
s &= \sum_{i=1}^{n-1} y_i'=3+1+1+2+3+1+2+3=16.
\end{align*}
\item Since $\Delta>s+q$, the decoder concludes that it belongs to the case (2b) where the deletion is not at the first position, i.e. $i\neq 1$, and $y_{i-1}+y_i>q=4$. 

\item Find the error location in $\by$. It can be observed that $7\times 4+\sum_{h=7}^9 y_i' =34 < \Delta=36$ while $8\times 4\sum_{h=8}^9 y_i' =37 > \Delta$. The decoder then concludes that the error in $\by$ is at the $h=7$ position, and hence, the error in $\bx$ is at $i=h+1=8$.  

\item To correct $\bx$, it inserts the symbol $0$ to the 8th position. 
\end{itemize}
\end{example} 

\begin{example}\label{exam:special}
We now consider a special case when the deleted symbol belongs to a run of identical symbols. Given $n=10, q=3, a=7$, and a codeword $\bx=01021{\color{blue}{222}}00\in {\rm Diff\_VT}_{7}({10;3})$. Suppose that we receive $\bx'=010212200$,  i.e. one can consider a deletion occurs at either $x_6$, or $x_7$, or $x_8$. We observe that $\by={\rm Diff}(\bx)=2111200200$ and $\by'={\rm Diff}(\bx')=211120200$. 
\begin{itemize}
\item The decoder computes $\Delta = a-{\rm Syn}(\by')=2 \ppmod{30}$ and $s=\sum_{j=1}^9 y_j'=9$. Since $\Delta<s$, the decoder concludes that it belongs to the case (2a). 
\item Observe that $\sum_{h=5}^9 y_i' = 4 > \Delta=2$ while $\sum_{h=8}^9 y_i' = 0 < \Delta=2$ (*). The first index where $\sum_{j=h}^{n-1} y_j'>\Delta$ is then $h=5$, i.e. the error in $\bx$ is at $i=h+1=6$. On the other hand, one may also select $h=6,7$ according to (*), i.e. the error is at $x_7$ or $x_8$, respectively. Nevertheless, we obtain the same codeword $\bx=0102122200$.
\end{itemize}
\end{example}

\begin{remark}
It is easy to show that our constructed codes ${\rm Diff\_VT}_{a}({n;q})$, from Construction 3, also support systematic linear-time encoder. The design is similar to the construction of the systematic encoder proposed by Tenengolts \cite{Tene:1984}. For message $\bx\in \Sigma_q^k$, the encoder appends the information of the VT syndrome of the differential vector of $\bx$ (of length $m+1=\ceil{\log_q n}+1$) into its suffix. In addition, there is a marker of length two, which serves as a separator between the data part and the redundancy part (refer to \cite{Tene:1984}). We illustrate the main idea of the encoder in Figure~\ref{fig:2encoders}a.
\end{remark}

\begin{figure}%
    \centering
    \subfloat[A systematic encoder for non-binary codes correcting a single deletion using the differential VT codes ${\rm Diff\_VT}_{a}({n;q})$. Here $m=\ceil{\log_q n}$. The combination ${\color{blue}{011}}$ at the end of the code sequence plays the role of the comma between transmitted sequences.The marker ${\color{blue}{pp}}$, where $p= x_k + 1 \ppmod{q}$ serves as separators between the data part and the redundancy part. Here, the output codewords do not belong to the same coset of the differential VT codes.]{{\includegraphics[width=8cm]{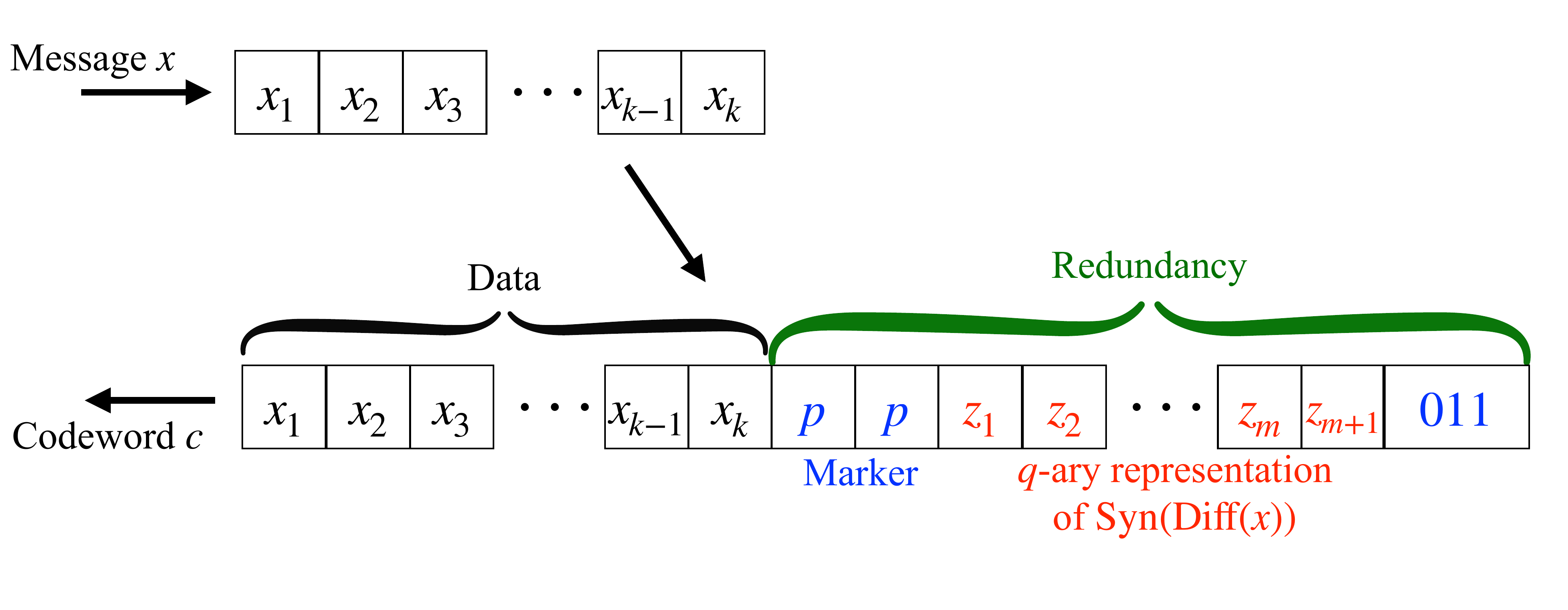} }} 
    \qquad
    \subfloat[An example of our designed linear-time encoder to encode arbitrary messages into the differential VT codes ${\rm Diff\_VT}_{a}({n;q})$ when $q=3$. In general, the VT syndrome ${\rm Syn(\by)}$ is computed in modulo $qn$ while each symbol is computed in modulo $q$. The set $S$ includes index $n$ and all powers of $q$. The message is of length $k=n-\ceil{\log_q n}-1$. Here, the output codewords belong to the same coset of the differential VT codes, i.e. the information of $a$ is known to the decoder.]{{\includegraphics[width=9cm]{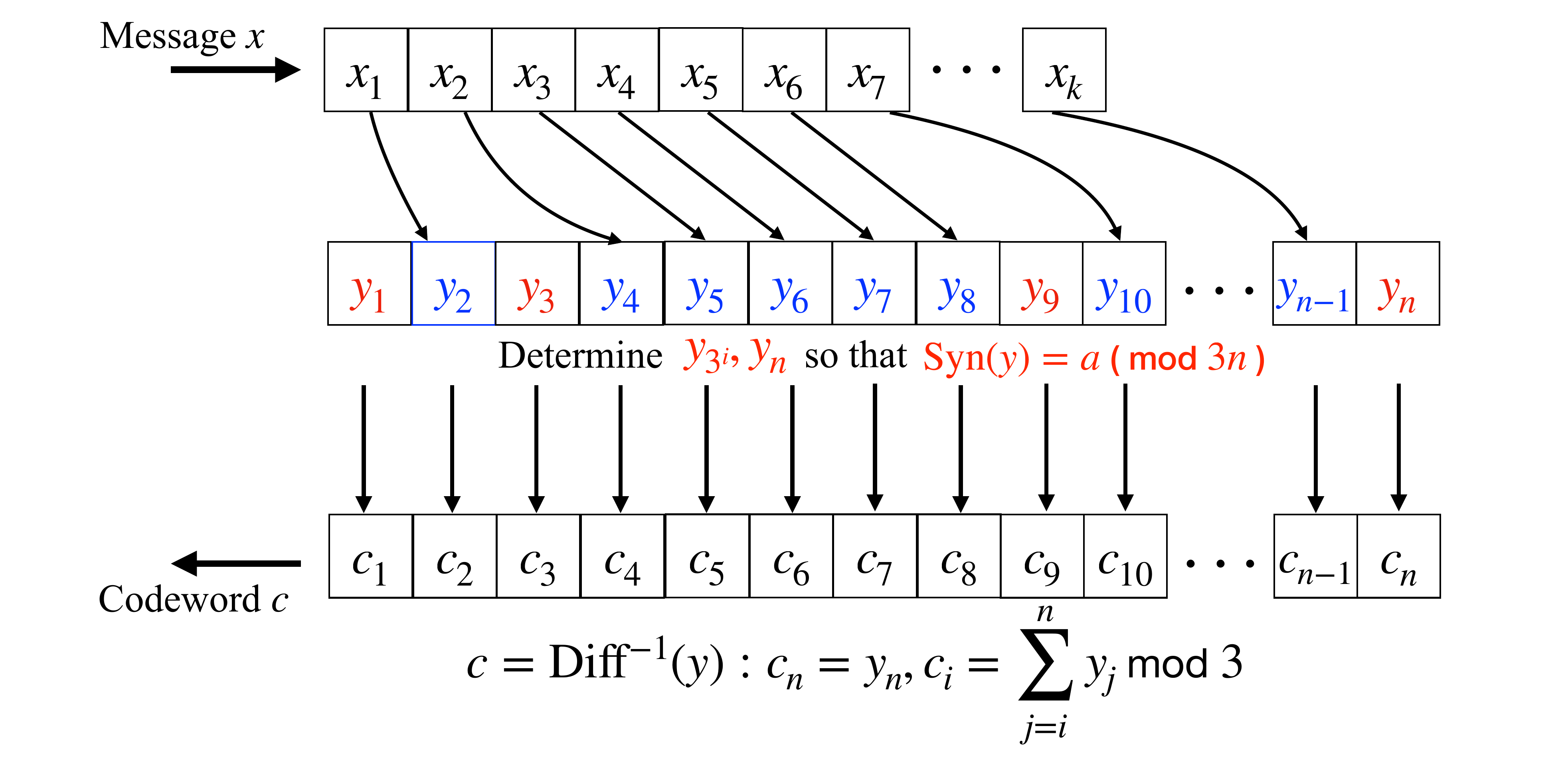} }}%
    \caption{Our proposed encoders for non-binary codes correcting a single deletion using the differential VT codes ${\rm Diff\_VT}_{a}({n;q})$. The construction of a systematic encoder is similar to the work proposed by Tenengolts \cite{Tene:1984}, both incur $\ceil{\log_q n}+3$ redundant symbols, while our best encoder (in Figure (b)) uses only $\ceil{\log_q n}+1$ redundant symbols.}%
    \label{fig:2encoders}%
\end{figure}



\subsection{More Efficient Encoder and Decoder of The Differential VT codes}

In this section, we present a linear-time encoder that encodes user data into the constructed differential VT codes ${\rm Diff\_VT}_{a}({n;q})$ with only $\ceil{\log_q n} + 1$ redundant symbols. 
\vspace{1mm}


\noindent{\bf The differential VT encoder} $\enc_{\rm Diff\_VT}$ 
\vspace{1mm}

{\sc Input}: $n,q,$ and $a\in \bbZ_{qn}$, a sequence $\bx\in \Sigma_q^k$, where $k\triangleq n-\ceil{\log_q n}-1$\\
{\sc Output}: $\bc \triangleq \enc_{\rm Diff\_VT}(\bx)\in {\rm Diff\_VT}_{a}({n;q})$\\[-2mm]
\begin{enumerate}[(I)]
\item Set $m\triangleq\ceil{\log_q n}$ and $S \triangleq \{q^{j-1} : j \in [m]\}\cup \{n\}$ and $I \triangleq [n]\setminus S$. In other words, the set $S$ includes the $n$th index and all the indices that are powers of $q$. 
\item Set $\by=y_1y_2\ldots y_n \in \Sigma_q^n$, where $\by|_I=\bx$ and $\by|_S=0$. In other words, the symbols in $\bx$ are filled into $\by$ excluding indices in $S$ (refer to Figure~\ref{fig:2encoders} (b)) and $y_j=0$ for $j\in S$. 
\item Compute the difference $a'\triangleq a-{\rm Syn}(\by) \ppmod{qn}$. 

In the next step, we modify $\by$, by setting suitable values for $y_j$ where $j\in S$, to obtain ${\rm Syn}(\by)=a \ppmod{qn}$. 
Since $0\le a'\le qn-1$, we find $\beta$, $0\le \beta<q-1$, to be the number such that $\beta n \le a' < (\beta+1)n$. 
\item The values for $y_j$ where $j\in S$ are set as follows. 
\begin{itemize}
\item Set $y_n=\beta$, and $a''=a'-\beta n < n$.
\item Let $z_{t-1}\ldots z_1z_0$ be the $q$-ary representation of $a''$. Clearly, since $a''<n$, the $q$-ary representation of $a''$ is of length at most $m=\ceil{\log_q n}$. We then have $a'' = \sum_{i=0}^{m-1} z_i q^i$. 
\item Set $y_{q^{j-1}}=z_{j-1}$ for $j\in [m]$. 
\end{itemize}
\item Set $\bc={\rm Diff}^{-1}(\by)$. In other words, we set $c_n=y_n$ and $c_i=\sum_{j=i}^n y_j \ppmod{q}$ for $1\le i\le n$. 

\item Output $\bc$.
\end{enumerate} 
\vspace{1mm}

\begin{theorem}\label{main-encoder} Our constructed encoder $\enc_{\rm Diff\_VT}$ is correct and has redundancy $\ceil{\log_q n}+1$ symbols. 
In other words, $\enc_{\rm Diff\_VT}(\bx)\in {\rm VT}_{a}({n;q})$ for all $\bx\in\Sigma_q^{n-\ceil{\log_q n}-1}$. 
\end{theorem}

\begin{proof}
We observe that the user message is of length $k=n-\ceil{\log_q n}-1$, and hence, the redundancy of the encoder is $\ceil{\log_q n} + 1$ symbols. It remains to show that  $\enc_{\rm Diff\_VT}(\bx)\in {\rm VT}_{a}({n;q})$ for all $\bx\in\Sigma_q^{k}$.
\vspace{1mm}

Suppose that $\bc=\enc_{\rm Diff\_VT}(\bx)$ for some $\bx\in\Sigma_q^k$. It suffices to show that 
${\rm Syn}({\rm Diff}(\bc))=a \ppmod{qn}.$ 
From Step (V) of the Encoder 2, $c={\rm Diff}^{-1}(\by)$, in other words, $\by={\rm Diff}(\bc)$. It remains to show that
${\rm Syn}(\by)=a \ppmod{qn}.$ 

Recall that from Step (I) of Encoder 2, $S \triangleq \{q^{j-1} : j \in [m]\}\cup \{n\}$ and $I \triangleq [n]\setminus S$. Therefore,  

\begin{small}
\begin{align*}
{\rm Syn}(\by) &= \sum_{j\in S}  jy_j+ \sum_{j\in I} jy_j \ppmod{qn}\\
&= \sum_{j\in [m]}  q^{j-1} y_j+ n y_n+ \sum_{j\in I} jy_j \ppmod{qn}\\
&= a'' + n \beta + (a-a') \ppmod{qn}\\
&= (a'-\beta n) + n \beta + a-a' \ppmod{qn}\\
&= a \ppmod{qn}. \qedhere
\end{align*} 
\end{small}
\end{proof}

We illustrate Encoder $\enc_{\rm Diff\_VT}$ via an example. 

\begin{example}
Consider $n=10, q=3$ and $a=0$. 
Then $m=\ceil{\log_3 10}=3$ and $k=10-3-1=6$.
Suppose that the message is $\bx=220011$ and we compute $\bc= \enc_{\rm Diff\_VT}(\bx)\in {\rm Diff\_VT}_{0}({10;3})$.

\begin{enumerate}[(I)]
\item Set $S=\{1,3,9,10\}$ and $I=\{2,4,5,6,7,8\}$.
\item The encoder first sets $\by={\color{red}{y_1}}2{\color{red}{y_3}}20011{\color{red}{y_9}}{\color{red}{y_{10}}}$.
It then sets $y_1=y_3=y_9=y_{10}=0$ to obtain $\by={\color{red}{0}}2{\color{red}{0}}20011{\color{red}{00}}$ and computes $a'=a- {\rm Syn}(\by)=0-27= 3 \ppmod{30}.$
\item Since $0<a'=3<10$, the encoder sets $\beta=0$ and $a''=a'=3$. It then sets $y_{10}=\beta=0$. 
\item The $3$-ary representation of $3$ is then $010$. 
Therefore, the encoder sets $y_1=0$, $y_3=1$, and $y_9=0$ to obtain $\by=0212001100$. We can verify that ${\rm Syn}(\by)=0  \ppmod{30}$.
\item The encoder outputs $\bc={\rm Diff}^{-1}(\by)=1121222100$. 
\end{enumerate}
\end{example}

For completeness, we state the corresponding decoder as follows.
\vspace{1mm}

\noindent{\bf The differential VT decoder} $\dec_{\rm Diff\_VT}$. Given $n,q,$ and $a\in \bbZ_{qn}$, $m\triangleq\ceil{\log_q n}$ and $k\triangleq n-m-1$. Given $\bc=\enc_{\rm Diff\_VT}(\bx)$ for some message $\bx\in \Sigma_q^k$, and suppose the decoder receives a sequence $\bc'$.  
\vspace{1mm}

{\sc Input}: $\bc'\in \Sigma_q^{n-1}\cup \Sigma_q^{n}\cup \Sigma_q^{n+1}$\\
{\sc Output}: $\bx=\dec_{\rm Diff\_VT}(\bc') \in\Sigma_q^k$\\[-2mm]
\begin{enumerate}[(I)]
\item The decoder follows the error-decoding procedure in Theorem 3 to obtain  $\bc\triangleq \dec_{\rm error}(\bc') \in \Sigma_q^n$.
\item Set $\by={\rm Diff}(\bc) \in \Sigma_q^n$, $y_i=c_i-c_{i+1} \ppmod{q}$ for $1\le i\le n-1$ and $y_n=c_n$.
\item Set $S \triangleq \{q^{j-1} : j \in [m]\}\cup \{n\}$ and $I \triangleq [n]\setminus S$.
\item Output $\bx=\by|_I \in \Sigma_q^k$.
\end{enumerate} 
\vspace{1mm}

To conclude this section, the efficiency of our proposed encoders, compared to previous works, is illustrated in Table~\ref{compare}.

\begin{table*}[h!]
\centering 
 \begin{tabular}{ |P{3cm}|P{3cm}| P{2.5cm}| P{2.5cm}| P{2cm}| P{2cm}| P{1cm}|}
 \hline
 Encoder &  Redundancy (in symbols) &  Encoding/Decoding Complexity &  Receiver Information on Code's Parameters &  Encoder Output &  Remark  \\[1ex]
 \hline
 Encoder proposed by Tenengolts \cite{Tene:1984} using ${\rm T}_{a,b}(n;q)$  & {\color{black}{$\ceil{\log_q n}+3+\ceil{\log_q 3}$}}   &  {\color{blue}{$O(n)$}} &  {\color{black}{not available}} & {\color{black}{not in ${\rm T}_{a,b}(n;q)$}} &  {\color{blue}{systematic}}  \\
 \hline
  Encoder proposed by Abroshan \et{} \cite{M:2018} using ${\rm T}_{a,b}(n;q)$  &  {\color{black}{ $> \log_q n+\log_2 n$}}   &  {\color{blue}{$O(n)$}} &  {\color{blue}{VT Syndrome and parity check}} &   {\color{blue}{in ${\rm T}_{a,b}(n;q)$}} &  {\color{blue}{systematic}}  \\
 \hline
Systematic encoder proposed in this work using ${\rm Diff\_VT}_{a}(n;q)$ (see Figure 1a)   &   {\color{black}{$\ceil{\log_q n}+3+\ceil{\log_q 3}$}}   &  {\color{blue}{$O(n)$}} & {\color{black}{not available}} & {\color{black}{not in ${\rm Diff\_VT}_{a}(n;q)$}} & {\color{blue}{systematic}}  \\
 \hline
Encoder $\enc_{\rm Diff\_VT}$ proposed in this work using ${\rm Diff\_VT}_{a}(n;q)$ (see Theorem~\ref{main-encoder} and Figure 1b)  &   {\color{blue}{$\ceil{\log_q n}+1$}}   &  {\color{blue}{$O(n)$}} &  {\color{blue}{VT Syndrome and parity check}} &  {\color{blue}{in ${\rm Diff\_VT}_{a}(n;q)$}} & {\color{black}{non-systematic}}  \\
\hline
\end{tabular}
\caption{Efficient encoders for $q$-ary codes correcting single deletion or insertion proposed in this work and and those in literature. For each design category, the most desirable option is highlighted in blue. Particularly, our proposed encoder $\enc_2$ incurs the least redundancy of $\ceil{\log_q n}+1$ symbols. Here, the receiver information on code's parameters plays an important role in error-detecting and error-correcting procedure. For example, it may provide more efficient basis for the design of segmented deletion/insertion correcting codes (see \cite{M:segment, Liu:2010, Cai:segment}).}
\label{compare}
\end{table*}

\section{Correcting a Burst of Fixed Length: The Differential Shifted VT Codes} 
For arbitrary fixed $t>1$, binary codes correcting a burst of exactly $t$ deletions were proposed in \cite{cheng-burst, binary-burst}. Recently, Schoeny \et{} \cite{schoenyDNA} extended the construction of binary codes in \cite{binary-burst} to the non-binary regime. To correct a burst of exactly $t$ deletions, for both the binary and non-binary cases, a common idea is to represent the codewords of length $n$ as a $t \times n/t$ codeword array, where $t$ divides $n$. Thus, for a codeword $\bx$, the codeword array $A_t(\bx)$ is formed by $t$ rows and $n/t$ columns. When $n/t$ is not an integer, one can append a sufficient number of bits/symbols $0$ into the suffix of each codeword (for example, see \cite{liveDNA}). In this work, for simplicity, we assume that $t$ divides $n$. Observe that a burst of $t$ deletions deletes in $\bx$ exactly one bit (in binary case) or one symbol (in a non-binary alphabet) from each row of the array $A_t(\bx)$.
\begin{equation*}
A_t(\bx)=
\left[
  \begin{array}{cccccc}
    x_1 & x_{t+1} & \cdots & x_{(j-1)t+1} & \cdots & x_{(n/t-1)t+1}\\
    x_2 & x_{t+2} & \cdots & x_{(j-1)t+2} & \cdots & x_{(n/t-1)t+2}\\
    \vdots & \vdots & \ddots & \vdots & \ddots & \vdots \\
      x_t & x_{2t} & \cdots & x_{jt} & \cdots & x_{n}\\
  \end{array}
\right].
\end{equation*}
Here, the $i$th row of the array is denoted by $A_t(\bx)_{\bf {\rm i}}$, and the $j$th column of the array is denoted by $A_t(\bx)^{\bf {\rm j}}$. We now briefly describe the coding methods in \cite{schoenyDNA} to correct a burst of exactly $t$ deletions in the general $q$-ary alphabet, $q\ge2$. The overall coding strategy in \cite{schoenyDNA} is split into two main parts. 

\begin{itemize}
\item The first row in the array belongs to a $q$-ary VT-code ${\rm T}_{a,b}({n;q})$ (refer to Construction~\ref{cons2}, Section II) that can correct a single error. In addition, such a code has an additional {\em run-length-limited (RLL)} property, that restricts the longest run of identical symbols to be at most $\ell=\ceil{\log_q n}+O(1)$. The authors also showed that for sufficiently large $n$, there exists a {\em runlength-limited encoder} which uses only one redundancy symbol to enforce such an RLL property. A similar design of such an encoder for binary codes was proposed in \cite{binary-burst}, that enforces binary codewords of maximum run length at most $\ceil{\log n} + 3$  with only one redundant bit (see \cite[Appendix B]{binary-burst}). The method is based on the {\em sequence replacement technique}. The idea can be extended to non-binary codes whose maximum runlength is at most $\ceil{\log_q n}+3$ with only one redundant symbol (for example, see \cite{Nguyen:2021}). 

\item Each of the remaining $(t- 1)$ rows in the array is then encoded using a modified version of the VT-code, 
which they refer as {\em shifted VT (SVT) code}. 
This code corrects a single deletion in each row provided the location of the error is known to be within $P$ consecutive positions. 
To obtain the desired redundancy, Schoeny \et{} also set $P=\ell+1=\ceil{\log_q n}+O(1)$.
\end{itemize}

\begin{lemma}[Nguyen \et{} \cite{Nguyen:2021}]\label{RLLencoder}
Given $n,q$, $\ell=\ceil{\log_q n} + 3$. There exist a linear-time encoder $\enc_{\ell\_{\rm RLL}}: \Sigma_q^{n-1} \to \Sigma_q^n$ and a corresponding decoder $\dec_{\ell\_{\rm RLL}}: \Sigma_q^{n} \to \Sigma_q^{n-1}$ such that the following conditions hold:
\begin{itemize}
\item For all $\bx\in\Sigma_q^{n-1}$, we have $\dec_{\ell\_{\rm RLL}}\circ\enc_{\ell\_{\rm RLL}}(\bx) = \bx$, 
\item If $\bc=\enc_{\ell\_{\rm RLL}}(\bx)$ then the maximum run of identical symbols in $\bc$ is at most $\ell$.
\end{itemize}
The redundancy of the encoder $\enc_{\ell\_{\rm RLL}}$ is one redundant symbol. 
\end{lemma}

\begin{definition}[\cite{schoenyDNA, binary-burst}] A {\em $P$-bounded single-deletion-correcting code} is a code in which the decoder can correct a single deletion given knowledge of the location of the deleted symbol to be within $P$ consecutive positions.
\end{definition}

Formally, the following results were provided by Schoeny \et{} \cite{schoenyDNA}. Recall that the signature vector of a $q$-ary vector $\bx$ of length $n$ is a binary vector $\alpha(\bx)$ of length $n-1$, where $\alpha(x)_i=1$ if $x_{i+1}\geq x_i$, and $0$ otherwise, for $i\in[n-1]$.


\begin{construction}[$q$-ary Shifted VT Codes \cite{schoenyDNA}]\label{shiftVTq}
For $0\leq a\le P$ and $0\le b<q$, $c \in \{0,1\}$, the $q$-ary shifted VT-code ${\rm SVT}_{a,b,c}(n,P,q)$ is defined as:
\begin{small}
\begin{align*}
{\rm SVT}_{a,b,c}(n,P,q) \triangleq \{\bx=(x_1,x_2,\ldots x_n): {\rm Syn}(\alpha(\bx)) = a \ppmod{P+1} \text{ and } \sum_{i=1}^{n} x_i = b \ppmod{q}  \text{ and } \sum_{i=1}^{n-1} \alpha(x)_i = c \ppmod{2}\}.
\end{align*}
\end{small}
\end{construction}

\begin{lemma}[Schoeny \et{} \cite{schoenyDNA}]\label{thm:svt}
The code ${\rm SVT}_{a,b,c}(n,P,q)$ is a P-bounded single deletion correcting code. 
\end{lemma}

\begin{theorem}[Schoeny \et{} \cite{schoenyDNA}]
There exists a $q$-ary code correcting a burst of exactly $t$ deletions whose number of redundancy symbols is at most 
\begin{equation*}
\log_q (n/t) + (t-1) \log_q (2(\log_q (n/t)+6) + t +1.
\end{equation*}
In term of bits, the redundancy is $\log n+(t-1) \log \log n+ O(t\log q)$ bits.  
\end{theorem}
As discussed in Section III, a drawback is the difficulty of enforcing VT syndrome over the signature vectors of the codewords. In this section, we extend the idea of the $q$-ary differential VT codes to construct the {\em $q$-ary differential shifted VT codes}, which are $P$-bounded single deletion correcting codes, but more importantly, they support more efficient encoding and decoding procedures.
\subsection{The Differential Shifted VT Codes}

\begin{construction}[{\bf $q$-ary Differential Shifted VT Codes}]\label{diffshiftVTq}
For $0\leq a< q(P+1)$ and $0\le b\le q$, the $q$-ary differential shifted VT-code ${\rm Diff\_SVT}_{a,b}(n;q,P)$ is defined as:
\begin{align*}
{\rm Diff\_SVT}_{a,b}(n;q,P) \triangleq \Big\{ \bx\in\Sigma_q^n: \text{ if } \by={\rm Diff}(\bx)\in \Sigma_q^n \text{ then } {\rm Syn}(\by) = a \ppmod{q(P+1)} \text{ and } \sum_{i=1}^{n} y_i = b \ppmod{q+1} \Big\}.
\end{align*}
\end{construction}

\begin{lemma}\label{thm:dsvt}
The code ${\rm Diff\_SVT}_{a,b}(n;q,P)$ is a P-bounded single deletion correcting code. 
\end{lemma}

\begin{proof}
Similar to the proof of Lemma~\ref{sumsymbol}, we have that if ${\rm Syn}({\rm Diff}(\bx)) = a \ppmod{q(P+1)}$ then we also have the parity check property, which consequently gives us the information of the deleted symbol: 
\begin{equation*}
x_1+x_2+\ldots+x_{n-1}+x_n = a \ppmod{q}.
\end{equation*}
Suppose that we receive the sequence $\bx'\in \B(\bx)$ of length $n-1$, the deleted symbol is $\gamma$, which can be determined from the parity check property, and the location of the error is within $L=[i, {i+1},\ldots, {i+P-1}]$. Now, assume that there are at least two locations in $L$ to insert the deleted symbol $\gamma$, i.e we obtain two different sequences $\bx_1$ (by inserting $\gamma$ at index $j_1$) and $\bx_2$ (by inserting $\gamma$ at index $j_2$) for some $i\le j_1 <j_2\le i+P-1$ so that all the code's constraints are satisfied, i.e.
\begin{align*}
\bx_1 &= ({\color{blue}{x_1', \ldots,  x_{j_1-1}'}}, \gamma, x_{j_1}', \ldots x_{j_2-1}', {\color{blue}{x_{j_2}', \ldots, x_{n-1}'}}), \text{ and} \\
\bx_2 &= ({\color{blue}{x_1', \ldots,  x_{j_1-1}'}}, x_{j_1}', \ldots x_{j_2-1}', \gamma, {\color{blue}{x_{j_2}', \ldots, x_{n-1}'}}).  
\end{align*} 

We now consider two cases. 

\noindent{\em Case 1.} If $j_1>1$. According to Lemma~\ref{trans-lemma1}, a deletion at symbol $x_j$ replaces $y_{j-1} y_{j}$ with $y_{j-1} + y_j \ppmod{q}$. From the information of $\by'={\rm Diff}(\bx')$, we can verify if $y_{j-1} + y_j\le q-1$ or $y_{j-1} + y_j\ge q$ as follows:
\begin{itemize}
\item If $y_{j-1} + y_j\le q-1$ then when $y_{j-1} y_{j}$ is replaced by $y_{j-1} + y_j \ppmod{q}$, there is no change in the sum of symbols in the differential vector. In other words, we must have 
\begin{equation}
\sum_{h=1}^{n-1} y'_h = b \ppmod{q+1}.
\end{equation}
\item On the other hand, if $y_{j-1} + y_j\ge q$ we observe that $y_{j-1} y_{j}$ is replaced by the new symbol $y_{j-1} + y_j - q$, and hence, 
\begin{equation}
\sum_{h=1}^{n-1} y'_h = b-q \ppmod{q+1}.
\end{equation}
\end{itemize}
Let $\bu={\rm Diff}(\bx_1)$ and $\bv={\rm Diff}(\bx_2)$. From (1) and (2) we must have $u_{j_1-1}+u_{j_1}=v_{j_2-1}+v_{j_2}$. On the other hand, we have $u_j=v_j$ for all $j\in [n] \setminus \{j_1-1,j_1,j_2-1,j_2\}$. Since ${\rm Syn}(\bu) = {\rm Syn}(\bv) \ppmod{q(P+1)}$, we have
\begin{align*}
{\color{blue}{\sum_{j=1}^{j_1-2} ju_j}} + \sum_{j=j_1-1}^{j_2} ju_j + {\color{red}{\sum_{j=j_2+1}^{n} ju_j}} &= {\color{blue}{\sum_{j=1}^{j_1-2} jv_j}} + \sum_{j=j_1-1}^{j_2} jv_j + {\color{red}{\sum_{j=j_2+1}^{n} jv_j}} \ppmod{q(P+1)},\text{ or} \\
 \sum_{j=j_1-1}^{j_2} ju_j  &=  \sum_{j=j_1-1}^{j_2} jv_j  \ppmod{q(P+1)},\text{ or} \\
{\color{teal}{(j_1-2)\Big( \sum_{j=j_1-1}^{j_2} u_j \Big)}}+ \sum_{j=1}^{j_2-j_1+2} ju_{j+j_1-2}  &=   {\color{teal}{(j_1-2)\Big( \sum_{j=j_1-1}^{j_2} v_j \Big)}}+ \sum_{j=1}^{j_2-j_1+2} jv_{j+j_1-2}  \ppmod{q(P+1)},\text{ or} \\
 \sum_{j=1}^{j_2-j_1+2} ju_{j+j_1-2}  &=  \sum_{j=1}^{j_2-j_1+2} jv_{j+j_1-2}  \ppmod{q(P+1)}.  
\end{align*} 
Thus, we obtain two sequences $\bx_3,\bx_4$ such that ${\rm Syn}({\rm Diff}(\bx_3)) = {\rm Syn}({\rm Diff}(\bx_4)) \ppmod{q(P+1)}$, where
\begin{align*}
\bx_3 &= (x_{j_1-1}', \gamma, x_{j_1}', \ldots x_{j_2-1}', x_{j_2}'), \text{ and} \\
\bx_4 &= (x_{j_1-1}',x_{j_1}',\ldots x_{j_2-1}', \gamma, x_{j_2}').  
\end{align*} 
Note that the length of $\bx_3$ and $\bx_4$ is $j_2-j_1+2 \le (i+P-1)-i+2 = P+1$, and hence we conclude that $\bx_3,\bx_4 \in {\rm Diff\_VT}^*_a(j_2-j_1+2,q)$ for some $0\le a< q(P+1)$ and such a code can correct a single deletion. On the other hand, we observe that $\bx''=(x_{j_1-1}', x_{j_1}', \ldots x_{j_2-1}', x_{j_2}')$ can be obtained from both $\bx_3$ and $\bx_4$ by deleting the symbol $\gamma$. We have a contradiction. It remains to consider the case when $j_1=1$. 

\noindent{\em Case 2.} If $j_1=1$ and $j_2 \le P$, i.e. $i=1$ and $L=[1,2,\ldots P]$.  Again, from (1) and (2), if $j_1=1$, we must have $y_{j_2-1}+y_{j_2} \le q-1$. 
Similarly, we obtain two sequences $\bx_3',\bx_4'$ of length at most $(P+1)$ such that ${\rm Syn}({\rm Diff}(\bx_3')) = {\rm Syn}({\rm Diff}(\bx_4')) \ppmod{q(P+1)}$, where
\begin{align*}
\bx_3' &= (\gamma, x_{1}', \ldots x_{j_2-1}', x_{j_2}'), \text{ and} \\
\bx_4' &= (x_1',\ldots x_{j_2-1}', \gamma, x_{j_2}').  
\end{align*} 
We have a contradiction. We conclude that there is at most one location to insert the deleted symbol $\gamma$ into $\bx'$, and thus, the constructed $q$-ary differential VT code ${\rm Diff\_SVT}_{a,b}(n;q,P)$ is a $P$-bounded single deletion correcting code. 
\end{proof}

\begin{remark}
We observe that our designed differential shifted VT codes ${\rm Diff\_SVT}_{a,b}(n;q,P)$ incur at most one more redundant symbol as compared to the $q$-ary shifted VT codes, proposed by Schoeny \et{} \cite{schoenyDNA} (refer to Construction~\ref{shiftVTq}). Particularly, the redundancy of a $q$-ary shifted VT code is $\log_q (P+1)+1+\log_q 2$ symbols while the redundancy of a differential shifted VT code in Construction~\ref{diffshiftVTq} is $\log_q (P+1)+1+\log_q (q+1)$ symbols. On the other hand, it provides an alternative simpler, and more efficient encoder (with the improvement of at least two redundant symbols as presented in Section III). 
\end{remark}



For completeness, we present an efficient encoder for the differential shifted VT codes ${\rm Diff\_SVT}_{a,b}(n;q,P)$, given arbitrary code parameters. Note that, in general, the value of $P$ is $\log_q n+O(1) =o(n)$. Given $q,n$, $0\leq a< q(P+1)$ and $0\le b\le q$. Set $m\triangleq\ceil{\log_q q(P+1)}$ and $k\triangleq n-m-2$. The message is of length $k$, and hence, the redundancy of our encoder is then $m+2= \ceil{{\log_q q(P+1)}}+2\approx \ceil{\log_q P}+3$. 
\vspace{2mm}

\noindent{\bf Differential SVT-Encoder} $\enc_{{\rm Diff\_SVT}}$.  
\vspace{1mm}

{\sc Input}: $n,q,P$ and $a\in \bbZ_{q(P+1)}$, a sequence $\bx\in \Sigma_q^k$, where $k\triangleq n-\ceil{\log_q q(P+1)}-2$\\
{\sc Output}: $\by \triangleq \enc_{{\rm Diff\_SVT}}(\bx)\in {\rm Diff\_SVT}_{a,b}(n;q,P)$\\[-2mm]
\begin{enumerate}[(I)]
\item Set index $i_0=2q(P+1)$ and $i_1=3q(P+1)$ where $i_0,i_1\le n$. Set $S \triangleq \{q^{j-1} : j \in [m]\} \cup \{i_0,i_1\}$ and $I \triangleq [n]\setminus S$.
\item Consider $\bc'\in \Sigma_q^n$, where $\bc'|_I=\bx$ and $\bc'|_S=0$. 
Compute the difference $a'\triangleq a-{\rm Syn}(\bc') \ppmod{q(P+1)}$. 
In the next step, we modify $\bc'$ to obtain a codeword $\bc$ with ${\rm Syn}(\bc)=a \ppmod{q(P+1)}$.
\item Let $z_{m}\ldots z_1z_0$ be the $q$-ary representation of $a'$ (since any number less than $q(P+1)$ has a representation of length at most $\ceil{\log_q q(P+1)}$).
In other words, $a' = \sum_{i=0}^{m} z_i q^i$. 
Then we set $c_{q^{j-1}}=y_{j-1}$ for $j\in [m]$.
\item Next, we set the symbols at the index $i_0=2q(P+1)$ and $i_1=3q(P+1)$ so that 
\begin{equation*}
c_{i_0}+c_{i_1}=b-\sum_{i\in [n]\setminus {\{i_0,i_1\}}} c_i\ppmod{q+1}.
\end{equation*}
\item Finally, we output $\by={\rm Diff}^{-1}(\bc)$.
\end{enumerate} 

\begin{theorem}\label{encoderSVT} The encoder $\enc_{{\rm Diff\_SVT}}$ is correct and has redundancy $\ceil{\log_q q(P+1)}+2$ symbols. 
In other words, $\enc_{{\rm Diff\_SVT}}(\bx)\in {\rm Diff\_SVT}_{a,b}({n;q,P})$ for all $\bx\in\Sigma_q^{k}$, where $k=n-\ceil{\log_q q(P+1)}-2$.
\end{theorem}

\begin{proof}
Suppose that $\by \triangleq \enc_{{\rm Diff\_SVT}}(\bx)$ for some $\bx\in\Sigma_q^k$. 
It suffices to show that 
\begin{equation*}
{\rm Syn}({\rm Diff}(\by))=a \ppmod{q(P+1)} \text{ and } \sum_{i=1}^{n} {\rm Diff}(\by)_i = b \ppmod{q+1}.
\end{equation*} 

From Step (V), we have $\by={\rm Diff}^{-1}(\bc)$. In other words, $\bc={\rm Diff}(\by)$, and it remains to show that ${\rm Syn}(\bc)=a \ppmod{q(P+1)}$. 
Recall that from Step (I), $S \triangleq \{q^{j-1} : j \in [m]\} \cup \{i_0=2q(P+1),i_1=3q(P+1)\}$ and $I \triangleq [n]\setminus S$. Therefore,  
\begin{align*}
{\rm Syn}(\bc) &= \sum_{j\in S}  jc_j+ \sum_{j\in I} jc_j \ppmod{q(P+1)}\\
&= \sum_{j\in [m]}  q^{j-1} c_j+ 2q(P+1) c_{2q(P+1)} + 3q(P+1) c_{3q(P+1)}+ \sum_{j\in I} jc_j \ppmod{q(P+1)}\\
&= a'+ 0 + 0+(a-a') \ppmod{q(P+1)}\\
&= a \ppmod{q(P+1)}. 
\end{align*} 
In addition, from Step (IV), we have $c_{i_0}+c_{i_1}=b-\sum_{i\in [n]\setminus {\{i_0,i_1\}}} c_i\ppmod{q+1}$, and hence, it implies that $\sum_{i=1}^{n} c_i = b \ppmod{q+1}$ or $\sum_{i=1}^{n} {\rm Diff}(\by)_i = b \ppmod{q+1}.$
\end{proof}

\begin{remark} We observe that reserving only one redundant symbol for the parity check constraint is not sufficient since the constraint is over modulo $(q+1)$. Similar to the construction of the differential VT decoder $\dec_{\rm Diff\_VT}$, one can easily obtain a corresponding differential shifted VT decoder $\dec_{\rm Diff\_SVT}$. We skip the detailed construction of such a decoder. 
\end{remark}

\subsection{Codes Correcting a Burst of $t$ Deletions with Efficient Encoder} 

We now present a construction of non-binary codes correcting a burst of $t$ deletions, and the coding method is based on the differential VT codes and the differential shifted VT codes as presented in earlier sections. Recall that we represent the codewords of length $n$ as $t \times n/t$ codeword arrays, where $t$ divides $n$. Given $\ell\ge 1$, a code $\C$ is called $\ell$-runlength limited if the maximum run of identical symbols in every codeword in $\C$ is at most $\ell$. 

\begin{construction}[{\bf $q$-ary $t$-burst-deletion correcting codes}]\label{t-burst-code}
Given $n,q$. Set $\ell=\ceil{\log_q n/t}+3, P=\ell+1$. For $0\le a_1 <q(n/t)$, $0\leq a_2< q(P+1)$ and $0\le b\le q$, let $\C_{a_1,a_2,b}(n,t;q)$ be a set:
\begin{align*}
\C_{a_1,a_2,b}(n,t;q)= \{\bx\in \Sigma_q^n:  &\text{ The first row: } A_t(\bx)_{\bf {\rm 1}} \in {\rm Diff\_VT}_{a_1}(n/t;q) \text{ and } A_t(\bx)_{\bf {\rm 1}} \text{ is $\ell$-runlength limited,} \\
&\text{ The other rows: }  A_t(\bx)_{\bf {\rm i}} \in {\rm Diff\_SVT}_{a_2,b}(n/t;q,P) \text{ for } 2\le i\le t\}.
\end{align*}
\end{construction}

\begin{theorem}\label{thm:t-burst-qary}
The code $\C_{a_1,a_2,b}(n,t;q)$ from Construction~\ref{t-burst-code} can correct a burst of $t$ deletions, and the redundancy is 
\begin{equation*}
\log_q (n/t) + (t-1)\log_q \log_q (n/t)+ O(t) \text{ (symbols)}.
\end{equation*}
In terms of bits, the redundancy is at most $\log n+(t-1) \log \log n+ O(t\log q)$ bits.  
\end{theorem}

\begin{proof}
The error-decoding procedure is similar to the construction of Schoeny \et{} \cite{schoenyDNA}. Since the first row belongs to a differential VT code, the decoder can recover the first row $A_t(\bx)_1$. In addition, since the maximum run of identical symbols in the first row is at most $\ell$, we can locate the error in each of the other rows to be within at most $P=\ell+1$ positions. Furthermore, since each of the remaining rows belongs to a differential shifted VT code, the decoder can recover each row accordingly. 

It remains to compute the redundancy of our constructed code. The redundancy used for the first row in the array $A_t(\bx)$ is 
\begin{equation*}
r_1= \log_q (qn/t) + 1 = \log_q (n/t) + 2 \text{ (symbols)}. 
\end{equation*}
Here, $\log_q (qn/t)$ symbols are used to encode the differential VT code while one additional symbol is to enforce the runlength-limited constraint. On the other hand, the redundancy used for each of the other $(t-1)$ rows in the array $A_t(\bx)$ is
\begin{equation*}
r_i= \log_q (q(P+1)) + \log_q (q+1) = \log_q (\ceil{\log_q n/t}+5) + 1 + \log_q (q+1) = \log_q \log_q (n/t) + O(1) \text{ (symbols) for } 2\le i\le t. 
\end{equation*}
Thus, the total redundancy is $\log_q (n/t) + (t-1)\log_q \log_q (n/t)+ O(t)$ symbols or $\log n+(t-1) \log \log n+ O(t\log q)$ bits.
\end{proof}

To conclude this subsection, we provide a linear-time encoder for such a code $\C_{a_1,a_2,b}(n,t;q)$ with given arbitrary code parameters. Observe that for $2\le i\le t$, the $i$th row $A_t(\bx)_i$ can be encoded/decoded independently by using the differential SVT-Encoder $\enc_{\rm Diff\_SVT}$, since there is no joint constraint among these rows. The redundancy to encode each of these rows is then $\ceil{\log_q q(P+1)}+2 \approx \ceil{\log_q P}+3$ for any $P$. On the other hand, to encode the first row $A_t(\bx)_1$, we need to enforce the runlength-limited constraint with the differential VT syndrome property. Recall that the encoder $\enc_{\rm Diff\_VT}$ for a differential VT code of length $n$ (as presented in Section III-B) uses only $\ceil{\log_q n}+1$ redundant symbols. 


\begin{lemma}\label{newRLL}
Given $n,q$. Set $\ell'=\ceil{\log_q n}+3$, and $k=n-\ceil{\log_q n}-2$. For an arbitrary sequence $\bx \in \Sigma_q^k$, suppose that $\by=\enc_{\ell'\_{\rm RLL}}(\bx) \in \Sigma_q^{k+1}$ and $\bc=\enc_{\rm Diff\_VT}(\by) \in {\rm Diff\_VT}_a(n;q)$. We then have the maximum run of identical symbols in $\bc$ is at most $\ell=2\ceil{\log_q n}+5$. 
\end{lemma}

\begin{proof}
Note that the differential VT encoder $\enc_{\rm Diff\_VT}$ of a code of length $n$ uses only $\ceil{\log_q n}+1$ redundant symbols at predetermined positions. Therefore, if the maximum run of identical symbols in $\bc=\enc_{\rm Diff\_VT}(\by)$ is at least $\ell+1=2\ceil{\log_q n}+6$, in other words, ${\rm Diff}(\bc)$ has at least $2\ceil{\log_q n}+5$ consecutive zeros (by definition of a differential vector), then the sequence $\by$ (before inserting $\ceil{\log_q n}+1$ redundant symbols) has a run of at least  $2\ceil{\log_q n}+5-{\ceil{\log_q n}}-1=\ceil{\log_q n}+4$ zeros. We have a contradiction since $\by$ is $\ell'$-runlength limited. 
\end{proof}

According to Lemma~\ref{newRLL}, to construct a $t$-burst encoder, we can set the value of $P$ to be $P=\ell+1=2\ceil{\log_q n}+6$, and amend the differential shifted VT code in the last $(t-1)$ rows and the corresponding encoder for such codes. For completeness, we present the detailed construction of a $t$-burst encoder as follows. 
\vspace{1mm}

\noindent{\bf Input.} Given $q,n$, $\ell'=\ceil{\log_q n}+3$, $\ell=2\ceil{\log_q n}+5$, $P=\ell+1=2\ceil{\log_q n}+6$, $0\le a_1 <q(n/t)$, $0\leq a_2< q(P+1)$ and $0\le b\le q$. The message $\bx\in\Sigma_q^k$ is of length 
\begin{equation*}
k=\Big( \underbrace{n/t-\ceil{\log_q n/t}-2}_{\text{ first row encoding }} \Big)+(t-1) \Big( \underbrace{n/t-\ceil{\log_q q(P+1)}-2}_{i\text{th row encoding, } 2\le i\le t} \Big) = n-\ceil{\log_q n/t}- (t-1) (\ceil{\log_q q(P+1)}+2)-2.
\end{equation*}

We observe that for $P=2\ceil{\log_q n}+6$, the total redundancy is then $\log_q (n/t) + (t-1)\log_q \log_q (n/t)+ O(t)$ symbols or $\log n+(t-1) \log \log n+ O(t\log q)$ bits.
\vspace{1mm}

\noindent{\bf $t$-Burst-Encoder} $\enc_{t\_{\rm burst}}$.  
\vspace{1mm}

{\sc Input}: Given $n,q$, and a sequence $\bx\in \Sigma_q^k$, where $k$ is defined above\\
{\sc Output}: $\by \triangleq \enc_{t\_{\rm burst}}(\bx)\in \C_{a_1,a_2,b}(n,t;q)$\\[-2mm]
\begin{enumerate}[(I)]
\item Suppose that $\bx=\bx_1\bx_2\ldots \bx_t$, where $\bx_1$ is the first $(n/t-\ceil{\log_q n/t}-2)$ symbols in $\bx$, and for $2\le i\le t$, $\bx_i$ is of length exactly $n/t-\ceil{\log_q q(P+1)}-2$. Set $k_1=n/t-\ceil{\log_q n/t}-2$ and $k_2=n/t-\ceil{\log_q q(P+1)}-2$. 
\vspace{1mm}

\item {\bf Encoding the first row in $A_t(\bx)$:} 
\begin{itemize}
\item Obtain $\bx_1'=\enc_{\ell'\_{\rm RLL}}(\bx_1) \in \Sigma_q^{k_1+1}$
\item Obtain $\by_1=\enc_{\rm Diff\_VT}(\bx_1') \in {\rm Diff\_VT}_{a_1}(n/t;q)$
\end{itemize}

\item {\bf Encoding the $i$th row in $A_t(\bx)$:} for $2\le i\le t$, we use the differential shifted VT encoder to obtain 
\begin{equation*}
\by_i =\enc_{\rm Diff\_SVT}(\bx_i) \in {\rm Diff\_SVT}_{a_2,b}(n/t;q,P).
\end{equation*} 

\item Finally, we output $\bc=\by_1||\by_2 || \ldots || \by_t$ (the interleaved sequence of $\by_1,\by_2, \ldots, \by_t$).
\end{enumerate} 

The following result is then immediate. 

\begin{theorem}
The encoder $\enc_{t\_{\rm burst}}$ is correct. In other words, the output codewords belong to $\C_{a_1,a_2,b}(n,t;q)$ that is capable of correcting a burst of $t$ deletions. The redundancy of the encoder is $\log_q (n/t) + (t-1)\log_q \log_q (n/t)+ O(t)$ symbols or $\log n+(t-1) \log \log n+ O(t\log q)$ bits. 
\end{theorem}

\section{Correcting a Burst of Variable Length} 


In this section, we focus on the case $t=2$, i.e. when there are at most two deletions. We first review the coding method of Wang \et{} \cite{2del:2}. To correct a burst of at most two deletions, the authors represent the codewords of length $n$ as a ${\ceil{\log q}} \times n$ codeword arrays, where each symbol in $\Sigma_q$ is converted to its binary representation of length ${\ceil{\log q}}$. For a sequence $\bu\in \Sigma_q^n$, 
\begin{small}
\begin{equation*}
A(\bu)=
\left[
  \begin{array}{c}
    \bx_1 \\
    \bx_2 \\
    \vdots \\
      \bx_{\ceil{\log q}}\\
  \end{array}
\right]
=
\left[
  \begin{array}{cccc}
    x_{1,1} & x_{1,2} & \cdots  & x_{1,n}\\
    x_{2,1} & x_{2,2} & \cdots & x_{2,n} \\
    \vdots & \vdots & \ddots & \vdots \\
      x_{\ceil{\log q},1} & x_{\ceil{\log q},2} & \cdots  & x_{\ceil{\log q},n}\\
  \end{array}
\right].
\end{equation*}
\end{small}
Therefore, the $q$-ary sequence $\bu$ is converted to a binary matrix with $\ceil{\log q}$ rows and $n$ columns. 
Observe that a burst of up to two deletions in $\bu$ spans at most two consecutive columns in $A(\bu)$ and there is a burst of up to two deletions in each binary row. Similar to the case of correcting a burst of $t$ deletions (as discussed in Subsection IV-A), the overall coding strategy in \cite{2del:2} is split into two main parts. 

\begin{itemize}
\item The first row in the array belongs to a binary code that can correct a burst of at most two deletions, proposed by Levenshtein in 1967 in \cite{le:2del}, that has redundancy $\log n + 1$ bits for codewords of length $n$. In addition, such a code has an additional {\em pattern length limited (PLL)} property, that restricts the maximum length of any substring with period 2 (the repetition of two consecutive bits instead of identical bits/symbols as in the RLL property) to be at most $\ell=\ceil{\log n}+O(1)$. The authors also showed that the redundancy to enforce both constraints in the first row is at most $\log n+3$ (refer to Construction 4, Lemma 2, \cite{le:2del}). 

\item Each of the remaining $(\ceil{\log q}- 1)$ rows in the array belongs to a modified version of the binary shifted VT-code, which can correct a burst of at
most 2 deletions with the positional knowledge (within $P$ positions) after recovering the first row. 
To obtain the desired redundancy, Schoeny \et{} also set $P=\ell+1=\ceil{\log n}+O(1)$. The redundancy used in each of the remaining $(\ceil{\log q}- 1)$ rows is at most $\log \log n+O(1)$ bits. 
\end{itemize}

The total redundancy of the coding scheme in \cite{le:2del} is $\log n+\log q\log \log n+O(\log q)$ bits. In this work, we use the idea of the differential shifted VT codes to further reduce the redundancy to construct a code correcting a burst of at most two deletions. The major difference in our coding scheme is that we view each $q$-ary sequence of length $n$ as a matrix with only two rows and $n$ columns. The mapping is designed as follows. 
\vspace{1mm}

Given $q> 2$. Set $q'=\ceil{q/2}$. For each symbol $x \in \Sigma_q$, {\em the decomposition} of $x$ in $\Sigma_{q'}$ is $\tau(x) = (x_1,x_2)$ where $x_1 \in \{0,1\},x_2\in \Sigma_{q'}$ and $x=x_1q'+x_2$. For example, when $q=3$, we have $\tau(0) = (0,0), \tau(1) = (0,1), \tau(2) = (1,0),$ and when $q=6$, we have $\tau(0) = (0,0), \tau(1) = (0,1), \tau(2) = (0,2), \tau(3) = (1,0), \tau(4) = (1,1), \tau(5) = (1,2).$

For a  $q$-ary sequence $\bx$ of length $n$ where $\bx=(x_1,x_2,\ldots x_n)$, we view it as the following matrix:

\begin{equation*}
D(\bx)=
\left[
  \begin{array}{cccc}
    \tau(x_1) & \tau(x_2) & \cdots & \tau(x_n) \\
  \end{array}
\right]
=
\left[
  \begin{array}{cccc}
    x_{1,1} & x_{1,2} & \cdots  & x_{1,n}\\
    x_{2,1} & x_{2,2} & \cdots & x_{2,n} \\
  \end{array}
\right],
\end{equation*}
where the first row $D(\bx)_1=(x_{1,1},x_{1,2},\ldots, x_{1,n}) \in \{0,1\}^n$, the second row $D(\bx)_2=(x_{2,1},x_{2,2},\ldots, x_{2,n}) \in \Sigma_{q'}^n$, and finally, the $i$th column $\tau(x_i)=(x_{i,1}, x_{i,2})^T$ for $1\le i\le n$.

Our overall coding scheme is as follows. For the first row, we also use the binary codes proposed by Levenshtein in 1967 in \cite{le:2del} that can correct a burst of at most two deletions with the PLL constraint as proposed by Wang \et{} \cite{2del:2}. On the other hand, for the second row, which is a $q'$-ary sequence, where $q=\ceil{q/2}$, we then use the differential shifted VT codes to correct the error given the positional knowledge of the errors. Before presenting our main contribution, we summarize the result of Wang \et{} \cite{2del:2}, which is used in our construction for the first row. 

\begin{lemma}[Construction 3, Construction 5, Wang \et{} \cite{2del:2}]\label{wang:lemma}
There exists a linear-time encodable and decodable binary code correcting a burst of at most two deletions and the maximum length of any substring with period 2 is at most $\ceil{\log n}+5$, and the code's redundancy is at most $\log n+3$ bits. 
\end{lemma}

We now present our main construction of non-binary codes correcting a burst of at most two deletions with only $\log n+3\log \log n+O(\log q)$ redundant bits. For simplicity, suppose that $n$ is even.

\begin{definition}
For a sequence $\bx=(x_1,x_2,\ldots x_n) \in\Sigma_q^n$, given $i,k>0$, we define the $(i;s)$\_subsequence of $\bx$, denoted by $\bx_{(i;s)}$, and $\bx_{(i;s)}=\Big(x_i, x_{i+s}, x_{i+2s}, \ldots, x_{i+s\floor{(n-i)/s}}\Big)$. 
\end{definition}

We observe that when $i=s=1$, we have $\bx_{(1;1)}\equiv \bx=(x_1,x_2,\ldots x_n)$. 

\begin{construction}[{\bf $q$-ary codes correcting at most two deletions}]\label{le2del-SVT}
Given $n,q$. Let $\C_1$ be a code obtained from Lemma~\ref{wang:lemma}. Set $\ell=\ceil{\log n}+5, P=\ell+1$ and $q'=\ceil{q/2}$. For $\ba=(a_1,a_2,a_3)$ and $\bb=(b_1,b_2,b_3)$, where $0\le a_1,a_2,a_3 <q'(P+1)$, $0\leq b_1,b_2,b_3\le q'$, let $\C_{\ba,\bb}(n,\le2;q)$ be a set of all $q$-ary sequences of length $n$ such that for each codeword $\bc$ the following conditions hold: 
\begin{itemize}
\item For the first row $D(\bc)_1$, we must have $D(\bc)_1 \in \C_1$
\item For the second row $D(\bc)_2$, suppose that $\bx=D(\bc)_2=(x_1,x_2,\ldots,x_n)$, we must have:
\begin{align}
\bx&=(x_1,x_2,\ldots,x_n) \in {\rm Diff\_SVT}_{a_1,b_1}(n;q',P) \\
\bx_{(1;2)}&=(x_1,x_3,\ldots,x_{n-1}) \in {\rm Diff\_SVT}_{a_2,b_2}(n/2;q',P), \\
\bx_{(2;2)}&=(x_2,x_4,\ldots,x_n) \in {\rm Diff\_SVT}_{a_3,b_3}(n/2;q',P).
\end{align}
\end{itemize}

\end{construction}

\begin{theorem}\label{thm:le2-burst-qary}
The code $\C_{\ba,\bb}(n,\le2;q)$ from Construction~\ref{le2del-SVT} can correct a burst of at most two deletions, and there exist sufficient values of $a_1, a_2, a_3, b_1, b_2, b_3$ such that the redundancy of such a code $\C_{\ba,\bb}(n,\le2;q)$ is 
$\log n + 3\log \log n+ O(\log q) \text{ bits}.$
\end{theorem}

\begin{proof}
We first show that such a code from Construction~\ref{le2del-SVT} can correct a burst of at most two deletions by providing an error-decoding algorithm. Suppose that from the codeword $\bc\in \C_{\ba,\bb}(n,\le2;q)$, the received sequence is $\bc'$. Clearly, from the length of $\bc'$, we can conclude the number of errors that occurred. If the length of $\bc'$ is $n$ then there is no error.

If the length of $\bc'$ is $n-1$, we conclude that there is a single deletion. Consequently, both rows $D(\bc)_1$ and $D(\bc)_2$ suffer exactly one deletion. Since $D(\bc)_1 \in \C_1$, which is a binary code capable of correcting a burst of up to 2 deletions, we can recover $D(\bc)_1$ uniquely. Next, since the maximum length of any substring with period 2 in $D(\bc)_1$ is at most $\ceil{\log n}+5$, the maximum run of identical bits in $D(\bc)_1$ is also at most $\ell=\ceil{\log n}+5$. We then conclude the location of the other error in $D(\bc)_2$ to be within determined $P$ positions where $P=\ell+1$. We then use the constraint (3) from the construction that $\bx=D(\bc)_2=(x_1,x_2,\ldots,x_n) \in {\rm Diff\_SVT}_{a_1,b_1}(n;q',P)$ to correct the error in $D(\bc)_2$. 

If the length of $\bx'$ is $n-2$, we conclude that there is a burst of exactly two deletions. Consequently, both rows $D(\bc)_1$ and $D(\bc)_2$ suffer exactly two consecutive deletions. In addition, we conclude that each subsequence, $\bx_{(1;2)}$ or $\bx_{(2;2)}$, suffers exactly a single deletion. 
Since $D(\bc)_1 \in \C_1$, which is a binary code capable of correcting a burst of up to 2 deletions, we can recover $D(\bc)_1$ uniquely. Next, since the maximum length of any substring with period 2 in $D(\bc)_1$ is at most $\ceil{\log n}+5$, we then conclude the location of the other error in $D(\bc)_2$ to be within determined $P$ positions where $P=\ell+1$. We then use the constraint (4) from the construction that $\bx_{(1;2)}=(x_1,x_3,\ldots,x_{n-1}) \in {\rm Diff\_SVT}_{a_2,b_2}(n/2;q',P)$ to correct the error in  $\bx_{(1;2)}$. Similarly, we use the constraint (5) from the construction that $\bx_{(2;2)}=(x_2,x_4,\ldots,x_n) \in {\rm Diff\_SVT}_{a_3,b_3}(n/2;q',P)$ to correct $\bx_{(2;2)}$. 

Thus, the code $\C_{\ba,\bb}(n,\le2;q)$ from Construction~\ref{le2del-SVT} can correct a burst of at most two deletions. It remains to show the redundancy of our designed codes. According to Lemma~\ref{wang:lemma}, the redundancy used for the first row is at most $\log n+3$ bits. On the other hand, the redundancy for a differential shifted VT code is $\ceil{\log_{q'} q'(P+1)}+2$ symbols, or $\log P+O(\log q)$ bits (see Theorem~\ref{encoderSVT}). In our construction, $P=\ell+1=\ceil{\log n}+6$, and hence, the redundancy for the second row to enforce three constraints (3), (4), and (5) is at most $3\log \log n+O(\log q)$ bits. Consequently, there exist sufficient values of $a_1, a_2, a_3, b_1, b_2, b_3$ such that the redundancy of $\C_{\ba,\bb}(n,\le2;q)$ is $\log n+3\log \log n+O(\log q)$ bits. 
\end{proof}

\begin{remark}
The idea of Construction~\ref{le2del-SVT} can be extended to construct non-binary codes correcting a burst of up to $t$ deletions. We still view each $q$-ary sequence $\bc$ of length $n$ as a matrix with exactly two rows and $n$ columns $D(\bc)$. Similar to the work of Wang \et{} \cite{WANG:BURST}, the first row belongs to a binary code that is capable of correcting a burst of up to $t$ deletions (for example, refer to the work of Lenz \et{} \cite{len:burst} for an efficient design of such a code) with an additional constraint to restrict the location of errors. Particularly, in \cite{WANG:BURST}, the authors show that it is possible to locate the errors within $O(\log n)$ positions using the concept of {\em $(w,\delta)$-dense string}. We then design the constraints for the second rows as in Construction~\ref{le2del-SVT} to handle every single case of $s$ deletions for any $s\le t$. In general, we would need $k(k+1)/2$ such constraints, resulting in a redundancy of at most $k(k+1)\log \log n+o(\log q)=O(k^2 \log \log n)$ bits (see Example~\ref{burst3}). However, the encoding and decoding procedures are much more complicated than the case of a burst of at most two errors. We defer the study of codes correcting a burst of up to $t$ deletions and the design of efficient encoders for such codes to our future works. 
\end{remark}

\begin{example}\label{burst3} To correct a burst of at most three deletions, for each codeword $\bc$, the first row $D(\bc)_1$ belongs to a binary code, which is capable of correcting a burst of at most three deletions with an additional constraint to locate the errors within $O(\log n)$ positions. On the other hand, the constraints for the second row $\bx=D(\bx)_2=(x_1,x_2, \ldots, x_n)$ are as follows: 
\begin{align}
\bx&=(x_1,x_2,\ldots,x_n) \in {\rm Diff\_SVT}_{a_1,b_1}(n;q',P) \\
\bx_{(1;2)}&=(x_1,x_3,\ldots,x_{n-1}) \in {\rm Diff\_SVT}_{a_2,b_2}(n/2;q',P), \\
\bx_{(2;2)}&=(x_2,x_4,\ldots,x_n) \in {\rm Diff\_SVT}_{a_3,b_3}(n/2;q',P),\\
\bx_{(1;3)}&=(x_1,x_4,x_7,\ldots) \in {\rm Diff\_SVT}_{a_4,b_4}(n/3;q',P), \\
\bx_{(2;3)}&=(x_2,x_5,x_8,\ldots) \in {\rm Diff\_SVT}_{a_5,b_5}(n/3;q',P), \text{ and}\\
\bx_{(3;3)}&=(x_3,x_6,x_9\ldots) \in {\rm Diff\_SVT}_{a_6,b_6}(n/3;q',P).
\end{align}
We observe that when there is exactly one deletion, the constraint (6) is sufficient to correct the error in the second row $\bx=D(\bx)_2$. On the other hand, when there is a burst of two deletions, the decoder uses the constraints (7) and (8) to correct the errors in $\bx_{(1;2)}$ and $\bx_{(2;2)}$, accordingly. Similarly, when there is a burst of three deletions, the decoder uses the remaining constraints (9), (10), and (11) to correct the errors in $\bx_{(1;3)}$, $\bx_{(2;3)}$ and $\bx_{(3;3)}$, when each of them suffers from a single deletion. 
\end{example}

To conclude this section, we present an efficient encoder for non-binary codes correcting a burst of at most two deletions with $\log n+3\log \log n+O(\log q)$ redundant bits, which significantly improves on the redundancy $\log q \log n +O(\log q)$ bits of the encoder in \cite{WANG:BURST}. Recall that in Construction~\ref{le2del-SVT}, for each codeword $\bc$, the rows $D(\bc)_1$ and $D(\bc)_2$ can be encoded independently. While the construction for a binary code satisfying the constraints required in the first row was presented in \cite{WANG:BURST}, it remains to present an efficient encoding algorithm for the second row $D(\bc)_2$. 
\vspace{1mm}

Note that the redundancy used in the differential shifted VT encoder $\enc_{\rm Diff\_SVT}$ is $\ceil{\log_q q(P+1)}+2$ symbols (see Theorem~\ref{encoderSVT}), where $\ceil{\log_q q(P+1)}$ symbols are used to enforce the syndrome constraint and the other two symbols are used to enforce the parity check constraint. 

\begin{construction}\label{burst2enc}
Given $n,q,P$, where $q>2$. Set $k=n-3(\ceil{\log_q q(P+1)}+2)-7$, $m=\ceil{\log_q q(P+1)}$. We construct an encoder $\enc^{*}: \Sigma_q^k \to \Sigma_q^n$ as follows. Suppose that $k$ is even and for a sequence $\bx=(x_1,x_2,\ldots x_k) \in\Sigma_q^k$, we obtain:
\begin{small}
\begin{align*}
a_1&={\rm Syn}({\rm Diff}(\bx)) \ppmod{q(P+1)}, \text{ and } b_1=\sum_{i=1}^k {\rm Diff}(\bx)_i \ppmod{q+1}, \\
a_2&={\rm Syn}({\rm Diff}(\bx_{(1;2)})) \ppmod{q(P+1)}, \text{ and } b_2=\sum_{i=1}^{k/2} {\rm Diff}(\bx_{(1;2)})_i \ppmod{q+1}, \\ 
a_3&={\rm Syn}({\rm Diff}(\bx_{(2;2)})) \ppmod{q(P+1)}, \text{ and } b_3=\sum_{i=1}^{k/2} {\rm Diff}(\bx_{(2;2)})_i \ppmod{q+1}.
\end{align*} 
\end{small}
Let $\bu_1, \bu_2, \bu_3$ be the $q$-ary representation of length $m=\ceil{\log_q q(P+1)}$ of $a_1, a_2$, and $a_3$, respectively. On the other hand, let $\bv_1, \bv_2, \bv_3$ be the $q$-ary representation of length $2$ of $b_1, b_2$, and $b_3$, respectively. Recall the last symbol in $\bx$ is $x_k$ and suppose that the first symbol in $\bu_1$ is $\beta$. Let $\gamma$ be the smallest symbol in $\Sigma_q$ that is different from $x_k$ and $\beta$, and obtain a marker $M=(x_k,x_k,\gamma,\gamma,\gamma,\beta,\beta)$ of length 7. We then set $\enc^{*}(\bx) \equiv \bx M \bu_1 \bv_1 \bu_2 \bv_2 \bu_3 \bv_3 \in \Sigma_q^n$.
\end{construction}

\begin{theorem}
Let $\C=\Big\{\enc^{*}(\bx): \bx \in\Sigma_q^k\Big\}.$ We then have that $\C$ can correct a burst of at most two deletions given the knowledge of the location of the deleted symbols to be within P consecutive positions.
\end{theorem}

\begin{proof}
Suppose that $\bc=\enc^{*}(\bx)$ for some $\bx\in\Sigma_q^k$ and the decoder receives a sequence $\bc'$, which is obtained from $\bc$ via a burst of at most two deletions. Recall the construction of the marker $M=(x_k,x_k,\gamma,\gamma,\gamma,\beta,\beta)$ of length 7, hence, when there is a burst of at most two deletions, we must have $c'_{k}=x_k$, $c'_{k+3}=\gamma$ and $c'_{k+6}=\beta$. Therefore, given the received sequence $\bc'$, the decoder is able to get the information of $x_k$, the last symbol in $\bx$, the symbol $\gamma$, and finally $\beta$, which is the first symbol in $\bu_1$. Base on the information of the marker $M$, it is able to locate the burst of at most two deletions, whether in $\bx$, or in the marker $M$, or in the suffix $\bu_1 \bv_1 \bu_2 \bv_2 \bu_3 \bv_3$. 
\begin{itemize}
\item If the errors occur at the marker $M$ or in the suffix $\bu_1 \bv_1 \bu_2 \bv_2 \bu_3 \bv_3$, the decoder concludes that there is no error in $\bx$ and simply takes the prefix of $k$ symbols as the original sequence $\bx$. To recover the suffix $\bu_1 \bv_1 \bu_2 \bv_2 \bu_3 \bv_3$, it proceeds to recompute $a_1, a_2, a_3, b_1, b_2, b_3$ as in Construction~\ref{burst2enc}, and recover the suffix $\bu_1 \bv_1 \bu_2 \bv_2 \bu_3 \bv_3$. 
\item On the other hand, if the errors occur within the first $k$ symbols in $\bx$, the decoder concludes that there is no error in the suffix $\bu_1 \bv_1 \bu_2 \bv_2 \bu_3 \bv_3$. Based on the information of this suffix and given the knowledge of the location of the deleted symbols to be within $P$ consecutive positions, the decoder follows the error-decoding procedure in Lemma 5 (refer to the $q$-ary differential shifted VT codes, Construction 5) to correct the errors in $\bx$. 
\end{itemize}
In conclusion, the code $\C=\Big\{\enc^{*}(\bx): \bx \in\Sigma_q^k\Big\}$ can correct a burst of at most two deletions given the knowledge of the location of the deleted symbols to be within $P$ consecutive positions.
\end{proof}
The following result is then immediate. 

\begin{corollary}
There exists a linear-time encoder $\enc$ and a corresponding decoder $\dec$ for non-binary codes correcting a burst of at most two deletions (or two insertions) with redundancy $\log n+3\log \log n+O(\log q)$ bits. 
\end{corollary}

\section{Conclusion}

We have presented a new version of non-binary VT codes that are capable of correcting a single deletion or single insertion, providing an alternative simpler and more efficient encoder of the construction by Tenengolts \cite{Tene:1984}. Our construction is based on the {\em differential vector}, and the codes are referred to as the {\em differential VT codes}. In addition, we have provided linear-time algorithms that encode user messages into these codes of length $n$ over the $q$-ary alphabet for $q \ge 2$ with at most $\ceil{\log_q n}+1$ redundant symbols, while the optimal redundancy required is at least $\log_q n+\log_q (q-1)$ symbols. Our designed encoder reduces the redundancy of the best-known encoder of Tenengolts \cite{Tene:1984} by at least  $2$ redundant symbols or equivalently $2\log_2 q$ bits. 

Moreover, we have introduced the {\em $q$-ary differential shifted VT codes} to construct non-binary codes correcting a burst of deletions (or insertions). Particularly, when there are at most two errors, our designed codes incur $\log n+3\log \log n+ O(\log q)$ redundant bits, which improves a recent result of Wang \et{} \cite{2del:2} with redundancy $\log n+O(\log q \log \log n)$ bits for all $q\ge 8$. We have also presented an efficient encoder for codes correcting a burst of exactly $t$ deletions (or insertions) for arbitrary $t\ge 1$, while the design of the encoder for codes correcting a burst of variable length (when the length is up to $t$ for arbitrary $t>3$) is deferred to our future work.




\end{document}